\documentclass[article]{acmart}

\usepackage{booktabs} % For formal tables

\usepackage[ruled]{algorithm2e} % For algorithms
\usepackage{subcaption}

\usepackage{float}

\SetAlFnt{\small}
\SetAlCapFnt{\small}
\SetAlCapNameFnt{\small}
\SetAlCapHSkip{0pt}
\IncMargin{-\parindent}

\DeclareRobustCommand{\R}{\mathbb{R}}
\DeclareRobustCommand{\N}{\mathbb{N}}
\DeclareRobustCommand{\tp}{\mathsmaller T}

\let\originalleft\left
\let\originalright\right
\DeclareRobustCommand{\left}{\mathopen{}\mathclose\bgroup\originalleft}
\DeclareRobustCommand{\right}{\aftergroup\egroup\originalright}

\def\clap#1{\hbox to 0pt{\hss#1\hss}}

\def\mathclapinternal#1#2{\clap{$\mathsurround=0pt#1{#2}$}}

\def\mathclap{\mathpalette\mathclapinternal}

\DeclareRobustCommand{\abs}[1]{\left\lvert #1\right\rvert}

\DeclareRobustCommand{\norm}[1]{\left\lVert #1\right\rVert}

\DeclareRobustCommand{\set}[1]{\left\{ #1\right\}}
\DeclareRobustCommand{\setn}[1]{\{ #1\}}

\let\p=\paren

\let\pb=\parenb

\let\sp=\sqparen

\let\smat=\sqmatrix

\newcommand{\reply}[1]{{\color{black} #1}}

\setcopyright{acmlicensed}
\begin{document}
% Title portion. Note the short title for running heads
\title{%Towards a Holistic Management of Wireless Control Systems\\\note{
  Holistic Cyber-Physical Management for Dependable Wireless Control Systems%
}
\author{Yehan Ma}
\author{Dolvara Gunatilaka}
\author{Bo Li}
\author{Humberto Gonzalez}
\author{Chenyang Lu}
%delete by Yehan
%\orcid{1234-5678-9012-3456}
\affiliation{%
  \institution{Washington University in St. Louis}
  \city{St. Louis}
  \state{MO}
  \postcode{63130}
  \country{USA}}

%\affiliation{%
%  \institution{Washington University in St. Louis}
%  \department{Computer Science and Engineering}
%  \city{St. Louis}
%  \state{MO}
%  \postcode{63130}
%  \country{USA}
%  }

%\affiliation{%
%  \institution{Washington University in St. Louis}
%  \department{Electrical and Systems Engineering}
%  \city{St. Louis}
%  \state{MO}
%  \postcode{63130}
%  \country{USA}
%}

\begin{abstract}
Wireless sensor-actuator networks (WSANs) are gaining momentum in industrial process automation as a communication infrastructure for lowering deployment and maintenance costs.  In traditional wireless control systems the plant controller and the network manager operate in isolation, which ignore the significant influence of network reliability on plant control performance.  To enhance the dependability of industrial wireless control, we propose a holistic cyber-physical management framework that employs run-time coordination between the plant control and network management.  Our design includes a holistic controller that generates actuation signals to physical plants and reconfigures the WSAN to maintain desired control performance while saving wireless resources.  As a concrete example of holistic control, we design a holistic manager that dynamically reconfigures the number of transmissions in the WSAN based on online observations of physical and cyber variables.  We have implemented the holistic management framework in the Wireless Cyber-Physical Simulator (WCPS).  A systematic case study has been presented based on two 5-state plants sharing a 16-node WSAN.  Simulation results show that the holistic management design has significantly enhanced the resilience of the system against both wireless interferences and physical disturbances, while effectively reducing the number of wireless transmissions.
\end{abstract}


% To by decided--Yehan




\thanks{
  % This work is supported by the National Science Foundation,
  % under grant CNS-0435060, grant CCR-0325197 and grant EN-CS-0329609.
  Author's addresses: Y.~Ma, D.~Gunatilaka, B.~Li, {and} C.~Lu, Computer Science and Engineering Department, Washington University in St.~Louis; H.~Gonzalez, Electrical and Systems Engineering Department, Washington University in St.~Louis.
}

\maketitle

\renewcommand{\shortauthors}{Y. Ma et al.}

\section{Introduction}
\label{sec:introduction}

With the adoption of industrial wireless standards such as WirelessHART~\cite{WirelessHART_standard} and ISA100~\cite{ISA_standard}, wireless sensor-actuator networks (WSANs) are being deployed in process industries world wide.
However, existing WSAN in process industries are usually used for \emph{monitoring} applications.
There remain significant challenges in supporting feedback control systems over WSAN due to concerns about \emph{dependability} of wireless control systems (WCS).  
A \emph{wireless control system} (WCS) employs a WSAN as the communication infrastructure for one or more feedback control loops, where the sensors, controllers and actuators communicate over the WSAN.
Despite considerable efforts to enhance the reliability of industrial WSAN, data loss is inevitable in open and hash operating environments, which may lead to severe degradation of control performance, or even system instabilities.
A dependable WCS therefore must maintain system stability and acceptable control performance under both physical disturbance and wireless interference.
A WCS is particularly vulnerable when significant data loss coincides with the plant experiencing a poor physical state. 
The challenge to utilize the WSAN for feedback control prevents the process industries from exploiting the full potential of wireless technologies, forcing plants to maintain extensive wired infrastructure despite the existence of the WSAN.  
Therefore, it is critical to develop WCS that are dependable under challenging cyber and physical conditions. 

Although the control system performance is heavily influenced by WSAN reliability, traditionally the physical plant and the network are controlled \emph{separately} at run time. 
The physical plant is controlled by a controller designed based on certain assumptions about the communication network.
However, the unpredictable wireless conditions of a WSAN mean that the wireless network design goals cannot be guaranteed, leading to unsafe physical plant operations.
On the other hand, the network is managed to reduce data loss in a best effort fashion, without any knowledge of the current requirements of the control system.
Notably, the required level of network reliability depends on the physical states of the plant.
When the physical plant is in an unsafe state, it requires highly reliable communication.
Conversely, when the physical plant is in a safe steady state, it is more tolerant to data loss which may in turn allow the network to save network resource.

\begin{figure*}[t]
 	\centering
  \includegraphics[width=0.45\textwidth]{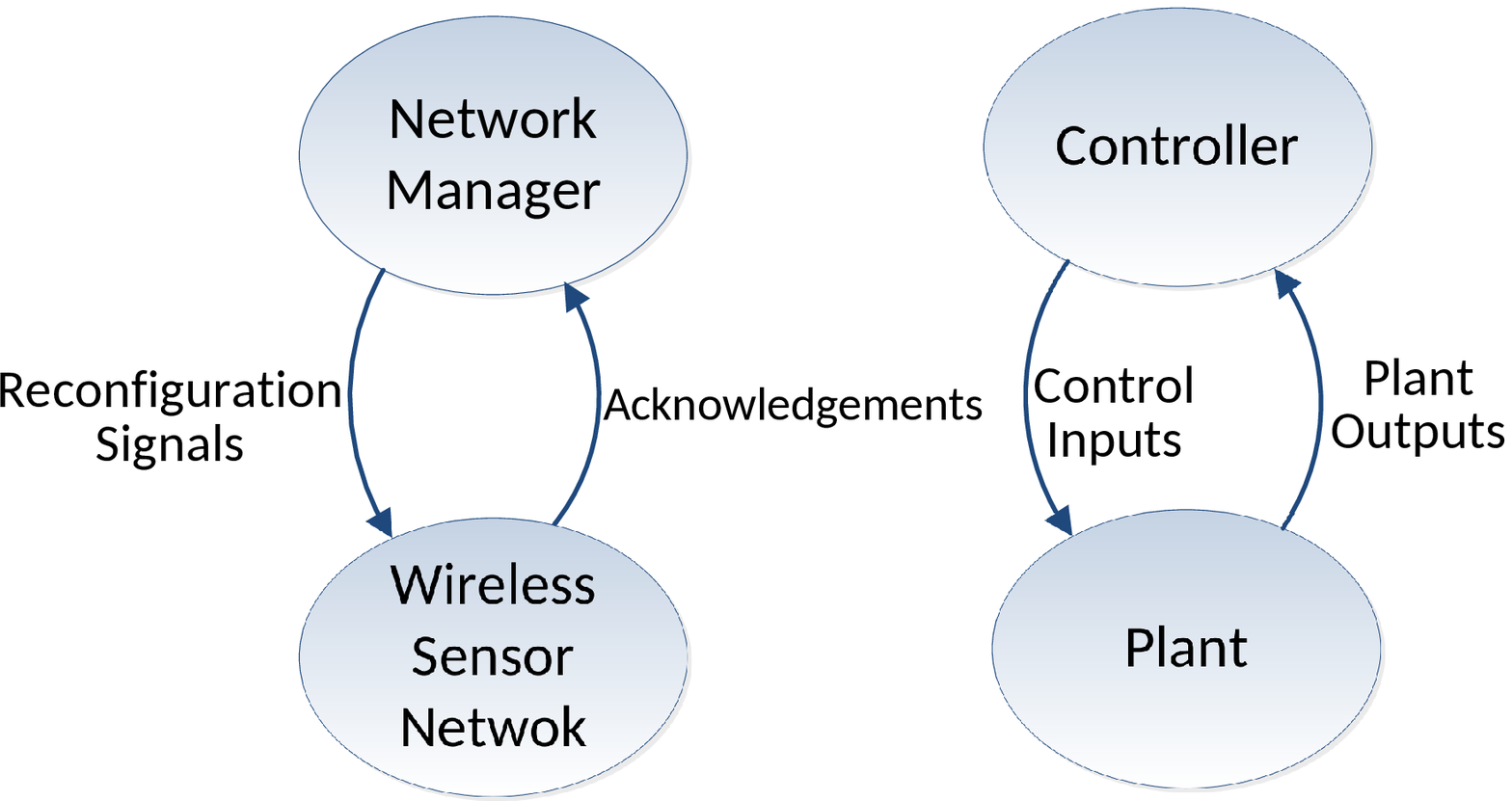}%
  \hspace{2em}%
  \includegraphics[width=0.45\textwidth]{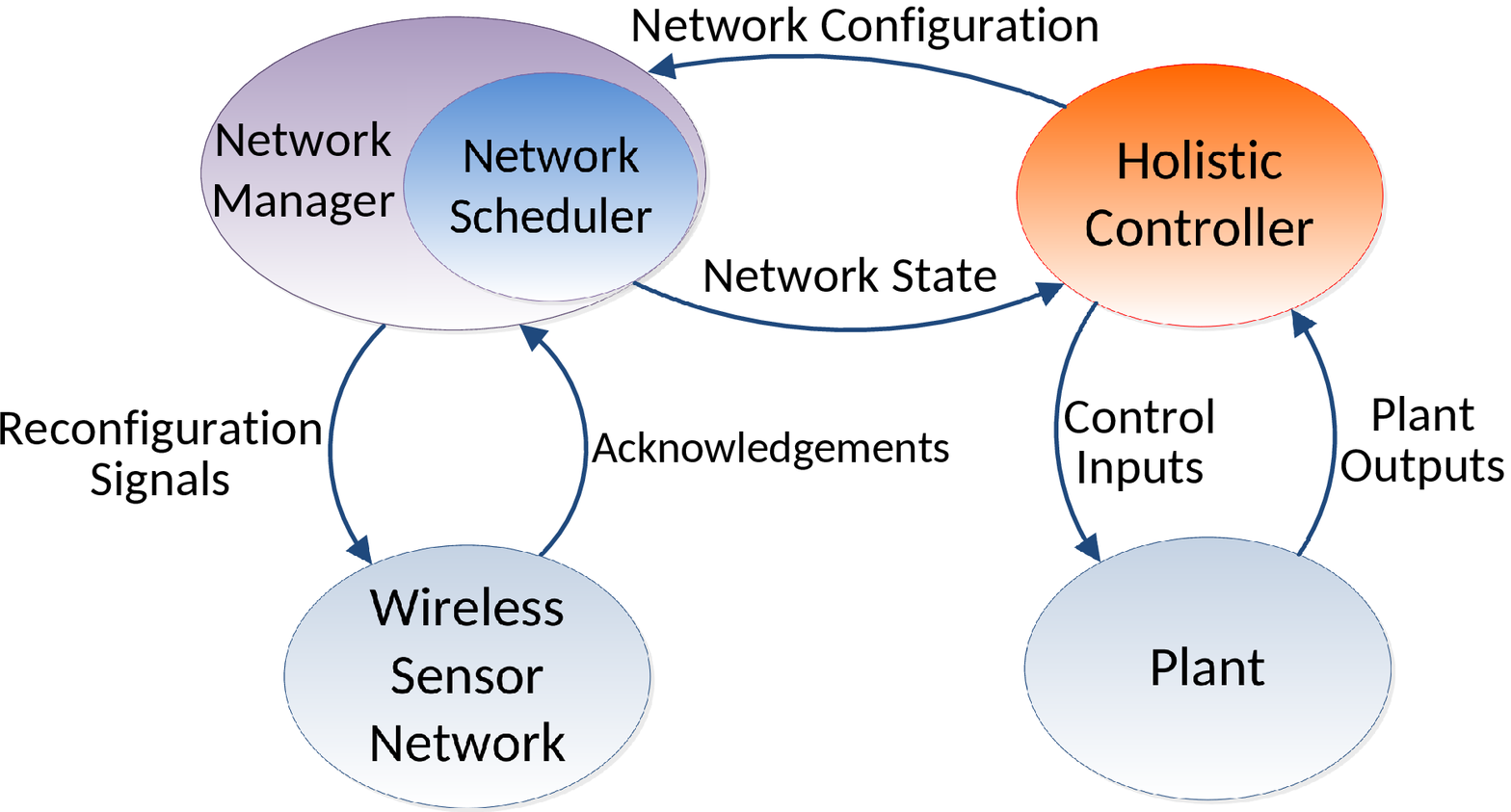}%
 	\caption{Traditional design (left) vs.\ holistic co-design (right) of wireless process control systems.}
 	\label{fig:holisticcontrol}
\end{figure*}

Building on this insight, we propose a \emph{holistic cyber-physical management} framework where the wireless network and the physical controller operate in a closed-loop fashion.
In this holistic management approach, the controller is endowed with the capability of modifying network configurations in addition to the physical plant itself, as shown in Fig.~\ref{fig:holisticcontrol}.
The controller that provides control commands to both the physical plant and the WSAN is regarded as the \emph{holistic controller}.  By coordinating the control and network management at run time, the holistic management approach can enhance the dependability of the WCS under both physical disturbance and wireless interference.  It can also improve the network efficiency when the plant state allows data loss.  

The contributions of this work are four-fold:
\begin{enumerate}
\item We propose the holistic management framework to enhance the dependability of WCS through coordinated control and network management;
\item We design a holistic controller based on a model predictive control scheme that simultaneously controls the physical plant while adjusting network configurations based on the state of the physical plant;
\item We present an example holistic management design that employs a network adaptation algorithm that dynamically adjust the number of retransmissions used to send certain actuation commands based on the physical plant conditions.
\item  We demonstrate through simulations that our holistic management scheme can enhance system dependability under both sensor disturbance and wireless interference, while avoiding allocating unnecessary retransmissions when not needed for the purpose of control.
\end{enumerate}

The rest of the paper is organized as follows:
Section~\ref{sec:related} describes our results in the context of related work.
Section~\ref{sec:wcs} presents the system architecture of our wireless networked control system.
Section~\ref{sec:controllerdesign} presents the design of our holistic controller and its adaptation algorithm.
Section~\ref{sec:wsn} presents an implementable solution for a run-time WSAN reconfiguration.
Section~\ref{sec:results} presents our experimental results.

\section{Related Work}
\label{sec:related}
%\note{Balance descriptions of stochastic properties, also do this in related works.}\\
%\note{Cite Redundant data transmission in control/estimation over lossy networks}\\
%\note{Cite On stochastic stability of packetized predictive control of non-linear systems over erasure channels~\cite{Pajic2012, ljevsnjanin2014packetized}}
%\note{~\cite{Branicky2002}}

Networked control systems are some of the best examples in the area of cyber-physical systems, and therefore have been extensively studied in the past decade.
As such, extensions of the traditional Kalman filter have been proposed for state estimation based on intermittent observations~\cite{Sinopoli2004, witrant2007predictive, irwin2010co}, \reply{enabling the use of lossy networks and delayed information in feedback control applications}.
New sampling methods, such as event-based and self-triggered control, have been developed to reduce communication traffic over wireless networks~\cite{Tabuada2007, Araujo2011}.
%\note{Assume that the network will have performance guarantees (no cyber attacks).}
%However, these results often overlook the natural loop that exists between network and control supervisors.
\reply{These results have addressed important problems related to cyber-physical co-design, and have validated the intuition that more network resources are needed when the physical states are away from a stable equilibrium point, while few network resources are needed when the physical states are close to a stable steady-state.
However, scheduling the number of packets that a network must transmit to control a system is only one side of the problem, since networks in general, and wireless networks in particular, can choose among many configurations (e.g., number of retransmissions, scheduling and routing) to deliver those packets~\cite{lureal, o2013ginseng, munir2010addressing, pottner2014constructing, han2011reliable}.
In other words, instead of abstracting the network as a transparent mechanism to transfer control information when needed, networked controllers should bidirectionally interact with the network manager as the dynamics of the physical system evolve.} That is, the control system and the network should be managed in a \emph{dynamic} and \emph{holistic} manner.

%\reply{Although maximizing packet delivery ratio with deadline constraints offers theoretically good control conditions~\cite{demirel2014modular, demirel2015wait}, in practice these efforts quickly reach a point where their impact in the performance of the physical plant is marginal, meanwhile they demand high efforts of the WSAN.
%This behavior is particularly present in the control of open-loop asymptotically stable plants, where even no network communication will result in a safe, albeit not optimal, performance.}

Dynamic stability results exist under several network configurations and communication conditions, including studies of delayed packet delivery and non-independent packet losses~\cite{montestruque2004stability, christofides2005control, Pajic2012, ljevsnjanin2014packetized}.
%\reply{However, most of the proofs are based on an assumption that network link failures can be regarded as Bernoulli distribution and some of them only work for single path, single hop, or single channel network control system.
\reply{However, stability guarantees for networked control systems typically come in the form of minimal requirements that a network must guarantee, such as bounds in data loss or in latency, thus disregarding cooperative approaches to adapt network conditions to physical plant performance and vice versa.
Even network protocols designed with control systems in mind, such as WirelessHART~\cite{WirelessHART_standard},
can only guarantee a specific level of reliability and performance under certain assumptions, thus potentially violating the sufficient conditions established in those theoretical results, especially when there are cyber or physical disturbances, or even malicious attacks.
Indeed, in the case of industrial wireless environments, these disturbances take the form of cross-protocol interference, physical obstacles, power failures, extreme weather, or sensor failures, among many others.}
% Therefore, the network link failure is hard to model in realistic.
% The network packet delivery ratio can be lower than $30\%$ due to cyber attacks, and models of plants will become inaccurate because of the physical attacks.
% \reply{Therefore, this paper offers an implementable network reconfiguration mechanism for wireless industrial process control, a mechanism that adapts to both wireless interferences and physical disturbances.}

The impact of retransmissions on control performance was studied in wireless control design. Previous works~\cite{demirel2014modular, demirel2015wait} showed that maximizing packet delivery ratio with deadline constraints offers theoretically good control conditions and explored the tradeoff between reliability and latency in control design. While those works focused on control design, our holistic management framework can dynamically change network configurations based on the states of the physical systems at run time.

From a networking perspective, several groups have worked on tailoring wireless network protocols for control systems applications.
For example, a cooperative MAC method was proposed to maintain control performance under unbounded delay, burst of packet loss, and ambient wireless traffic~\cite{ulusoy2011wireless}.
Bernardini and Bemporad have proposed a communication strategy between sensors and the controller that minimizes the data exchange over the wireless channel~\cite{bernardini2008energy}.
Also, several groups have considered specific scheduling and control schemes for closed-loop systems with stability guarantees~\cite{KTH-codesign, li_iccps13, Bai12:crosslayer, Schenato2007, Franceschelli2008, koutsoukos2008passivity}. Our effort is complementary to those works in that we aim to develop a holistic cyber-physical management framework for wireless control systems instead of developing new network protocols. 

In our previous studies~\cite{li2015incorporating, li2016wireless}, we have experimentally investigated the effect of harsh network conditions on control loops.
We have also developed specific routing and scheduling protocols to mitigate their effect.
In this paper we build upon our previous results, developing a holistic management scheme that instead of asking the network for a fixed minimum set of performance conditions, updates both control algorithm and network configurations, resulting in a robust and safe physical execution and an efficient network information flow.
As shown in Section.~\ref{sec:results}, our holistic management scheme increases the resilience of the closed-loop system, even in the presence of high wireless background noise and sensor malfunctions.
\section{A Wireless Control System}
\label{sec:wcs}

\reply{Since the working conditions of industrial plants are always too harsh for the controllers to operate, and since one industrial PC is supposed to control multiple loops, controllers are always located far away from the plants in wireless industrial process control deployments.}
Therefore, standard industrial wireless control systems use multi-hop networks, such as \reply{ISA~\cite{ISA_standard}}, WirelessHART~\cite{WirelessHART_standard}, and \reply{ZigBee}~\cite{Zigbee_standard}
to deliver information from a collection of sensors to a remote controller, and then back from the controller to the actuators in the plant.
In this paper, we adopt this general architecture, which we upgrade by using two modules to mitigate the impact of information loss in the wireless network: a state observer and an actuation signal buffer.
Below, we explain in detail how each of the components in our architecture interact to close the control feedback loop.

\subsection{Physical Plant and Controller}
\begin{figure}[tp]
	\centering
	\includegraphics[width=0.8\columnwidth]{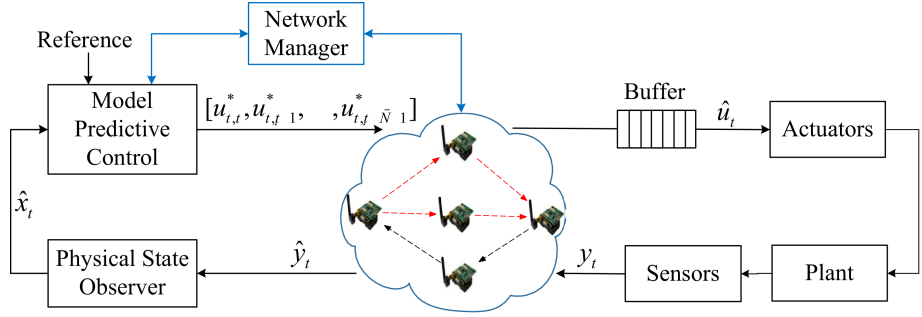}
	\caption{System Architecture.}
	\label{fig:wsn}
\end{figure}

Fig.~\ref{fig:wsn} shows the closed-loop system architecture that we consider in this paper, which closely follows to the architecture we considered in~\cite{li2016wireless}.
We assume that the plant is a linear time-invariant system, governed by the following difference equation:
\begin{equation}
  \label{eq:diffeq}
  x_{t+1} = A\, x_t + B\, u_t, \quad y_t = C\, x_t, \quad t \in \N,
\end{equation}
where $t$ is the time index, $x_t \in \R^n$ is the state vector, $u_t \in \R^m$ is input vector, $y_t \in \R^p$ is the output vector,
$A \in \R^{n \times n}$, $B \in \R^{n \times m}$, and $C \in \R^{p \times n}$.
We assume that the pair $(A,B)$ is controllable and that the pair $(A,C)$ is observable.

Sensor measurements $y_t$ are transmitted to a remote controller via the wireless network.
At time $t$, the wireless network delivers a vector $\hat{y}_t$, which is either equal to $y_{t-1}$ or is a \emph{no-data} signal if the packet is unable to be delivered within the time sampling deadline.
The vector $\hat{y}_t$ is delivered to an intermittent observation Kalman filter~\cite{Sinopoli2004}, which produces an estimated state vector $\hat{x}_t$ regardless of whether the wireless network was capable of delivering the sensing information.

At the core of the remote controller is a model predictive control scheme that at each time $t$ solves the following optimal control problem:
\begin{equation}
  \label{eq:mpc}
  \begin{aligned}
    V(\hat{x}_t) \!= \hspace{-7pt}\min_{\set{u_j}_{j=t}^{t+N-1}} &\sum_{j=t}^{\mathclap{t+N-1}} \p{x_j^\tp\, Q\, x_j + u_j^\tp\, R\, u_j} + x_{t+N}^\tp\, S\, x_{t+N},\\
    \text{subject to:}\quad
    &x_t = \hat{x}_t,\\
    &x_{j+1} = A\, x_j + B\, u_j,\\
    &x_j \in {\mathcal X},\ u_j \in {\mathcal U},\ j \in \set{t,\dotsc,t+N\!-\!1},\\
    &x_{t+N} \in {\mathcal X}_f,
  \end{aligned}
\end{equation}
where $N \geq 0$ is the \emph{time horizon}, $Q, S \in \R^{n \times n}$ are positive semi-definite, $R \in \R^{m \times m}$ is positive definite, ${\mathcal U} \subset \R^m$ is the input constraint set, and ${\mathcal X}, {\mathcal X}_f \subset \R^n$ are the state safety and final constraint sets, respectively.
We denote by $\setn{u_{t,j}^*}_{j=t}^{t+N-1}$ the optimal input signal at time $t$, which is the minimizer associated with value $V(\hat{x}_t)$ in~\eqref{eq:mpc}.
Note that the controller in~\eqref{eq:mpc} can also be used to control the system in~\eqref{eq:diffeq} around any state reference $\bar{x} \in {\mathcal X}$, satisfying:
\begin{equation}
  \label{eq:state_ref}
(A-I)\, \bar{x} + B\, \bar{u} = 0,\ \text{for some}\ \bar{u} \in {\mathcal U}.
\end{equation}
We assume that $\bar{x} \in {\mathcal X}$ and $\bar{u} \in {\mathcal U}$.

The MPC scheme in~\eqref{eq:mpc} sends the resulting sequence of optimal inputs, $\setn{u_{t,j}^*}_{j=t}^{t+N-1}$, over the wireless network.
Whenever the sequence successfully traverses the wireless network, it overwrites the old sequence stored in a buffer that periodically feeds the plant actuators.
Hence, if the packet containing the input sequence is successfully delivered at time $t$, then the actuator will apply the vector $u_{t,t}^*$.
If, instead, the packet is dropped, then the actuator will apply the vector $u_{t-1,t}^*$, and a similar procedure is repeated if consecutive actuation packets are dropped.
Thus, the buffer allows us to feed the actuator with an optimal (albeit potentially obsolete) input, even if $N$ packets are consecutively dropped, as explained in~\cite{li2016wireless}.

It is worth noting that while the state observer provides a robust and theoretically sound protection against loss of sensing information, the buffers delivering samples to the actuators are implemented following a heuristic approach.
Indeed, one would expect that if the MPC scheme is properly tuned and there are no external disturbances, then any two consecutive actuation signals, say $u_{t,t+1}^*$ and $u_{t+1,t+1}^*$, will not be very different.
The limit case occurs when the plant reaches steady-state and no external disturbances are applied to the plant, where the MPC scheme computes $N$ identical samples every time, i.e. $u_{t,j}^* = \bar{u}$ for each $j$.
Thus, in that case, our buffer heuristics allows us to withstand the loss of $N$ consecutive actuation packets without a control performance loss.
This situation is exemplified in both simulations in Section.~\ref{sec:results}.

\subsection{Wireless Network and Manager}
%\note{Explain why we use WirelessHART, TDMA protocol}
\label{sec:WNM}
We adopt a WirelessHART~\cite{WirelessHART_standard} architecture for our WSAN design, \reply{which is designed for applications in industrial wireless process automation by selecting a set of specific network features that enable timely and highly reliable communication.}
A WirelessHART network is a wireless multi-hop mesh network consisting of a number of battery-powered field devices connected to a gateway through access points.
The network is managed by a centralized network manager, usually collocated with the gateway.
The network manager collects topology information from the field devices, computes routes and transmission schedules, and disseminates the routing information and schedules among field devices.

\reply{WirelessHART adopts the IEEE 802.15.4 physical layer designed for low cost and low data rate communication.} Transmissions are scheduled based on a multi-channel Time Division Multiple Access (TDMA) protocol \reply{which can provide a deterministic and collision-free communication compared to CSMA/CA, and which works perfectly with periodic communication}.
Each time slot is $10\, \text{ms}$, which can accommodate a transmission and its acknowledgement.
For transmissions between pairs of nodes, a time slot can either be \emph{dedicated} or \emph{shared}.
In a \emph{dedicated} slot, only one sender is allowed to transmit.
In a \emph{shared} slot, more than one sender competes for one transmission opportunity.
WirelessHART networks operate on a $2.4\, \text{GHz}$ ISM band, and can use up to 16 channels, as defined in IEEE~802.15.4 physical layer standard.
Also, WirelessHART networks  adopt \emph{channel hopping} for channel diversity, periodically changing the communication channel according to a predetermined schedule.
At the network layer, the WirelessHART protocol supports two types of routings, namely, \emph{source} and \emph{graph} routing.
\emph{Source} routing provides a single path from source to destination, while \emph{graph} routing consists of multiple routes for each pair of source and destination.

As explained in Section.~\ref{sec:introduction}, the performance of the wireless network and the closed-loop control system are intertwined.
Among all the statistics one can measure regarding the performance of a wireless network, the packet delivery ratio (PDR) is at first sight the most significant for control applications, since a perfect PDR implies that all the information sent through the network is eventually delivered.
On the other hand, high PDR in multi-hop networks come associated with long delays \reply{due to redundant transmission and route diversity}, which can be longer than the information flow deadlines.
Therefore, in this paper, we propose a novel holistic controller, capable of balancing physical and wireless requirements while maintaining the stability of the plant.
The holistic controller will simultaneously compute actuation signals and command the network to update its configuration, as a function of current PDR measurements and worst-case state performance predictions.

Network reliability can be achieved through different means such as packet retransmission, route diversity, or channel diversity, among others.
Our design adopts a mechanism in which the number of transmissions of a certain route changes according to the conditions of a network and a physical plant.
In this paper, we avoid modifying more than one network configuration parameter to simplify our presentation and avoid an unnecessarily complex algorithm.

\begin{figure}[t]
	\centering
	\includegraphics[width=0.5\columnwidth]{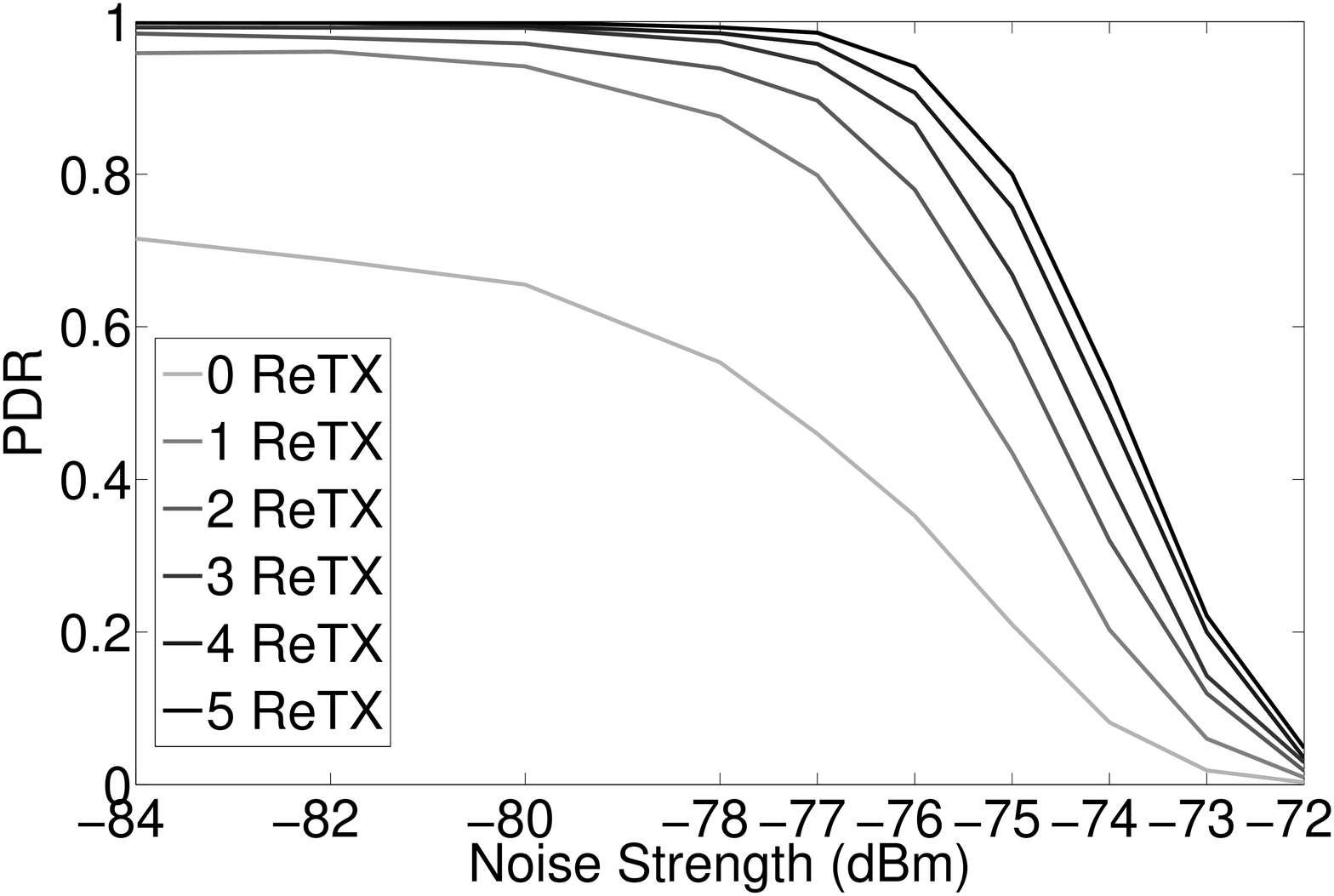}
	\caption{%
    Average PDR measured for different retransmission indices under different wireless background noise conditions.
  }
	\label{fig:re-tx}
\end{figure}

The transmission number ($\#TX$) is at the center of a tradeoff between reliability and \reply{network resources}, i.e.\ more transmissions lead to a higher delivery ratio at a cost of \reply{network resources}.
Fig.~\ref{fig:re-tx} depicts the relationship between end-to-end delivery ratio and the retransmission number ($\#ReTX$) under a 16-node WSAN. There is a diminishing return of PDR improvement as the $\#ReTX$ increases. \reply{(The settings of this set of experiments are the same with Section.~\ref{sec:results})}

\begin{figure}[t]
\centering
\begin{minipage}[b]{0.4\textwidth}
\centering
    \includegraphics[width=1\columnwidth]{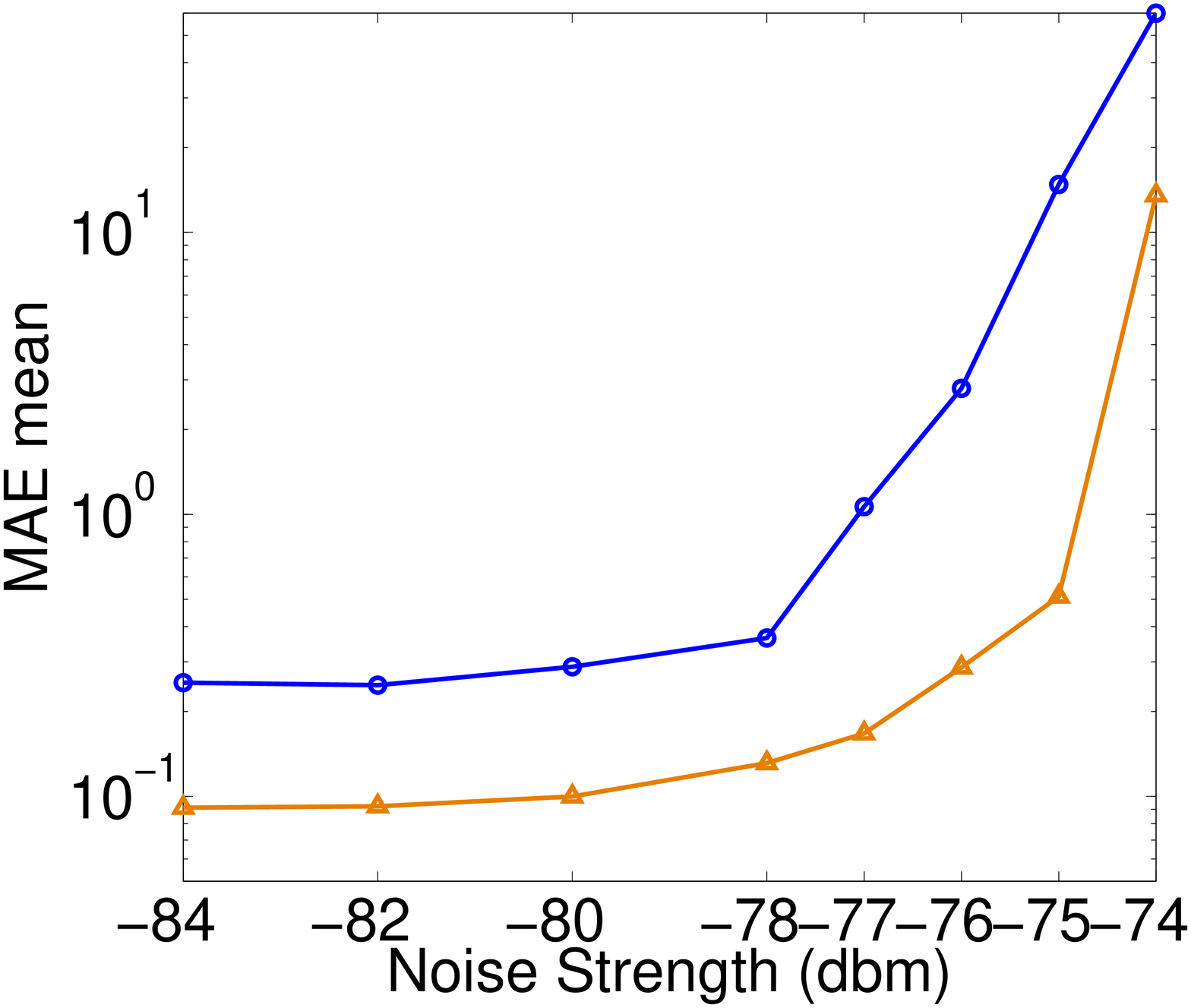}%
    \subcaption{Open-loop unstable plant}
    \label{fig:runstable}%
    \end{minipage}
    \begin{minipage}[b]{0.4\textwidth}
\centering
    \includegraphics[width=1\columnwidth]{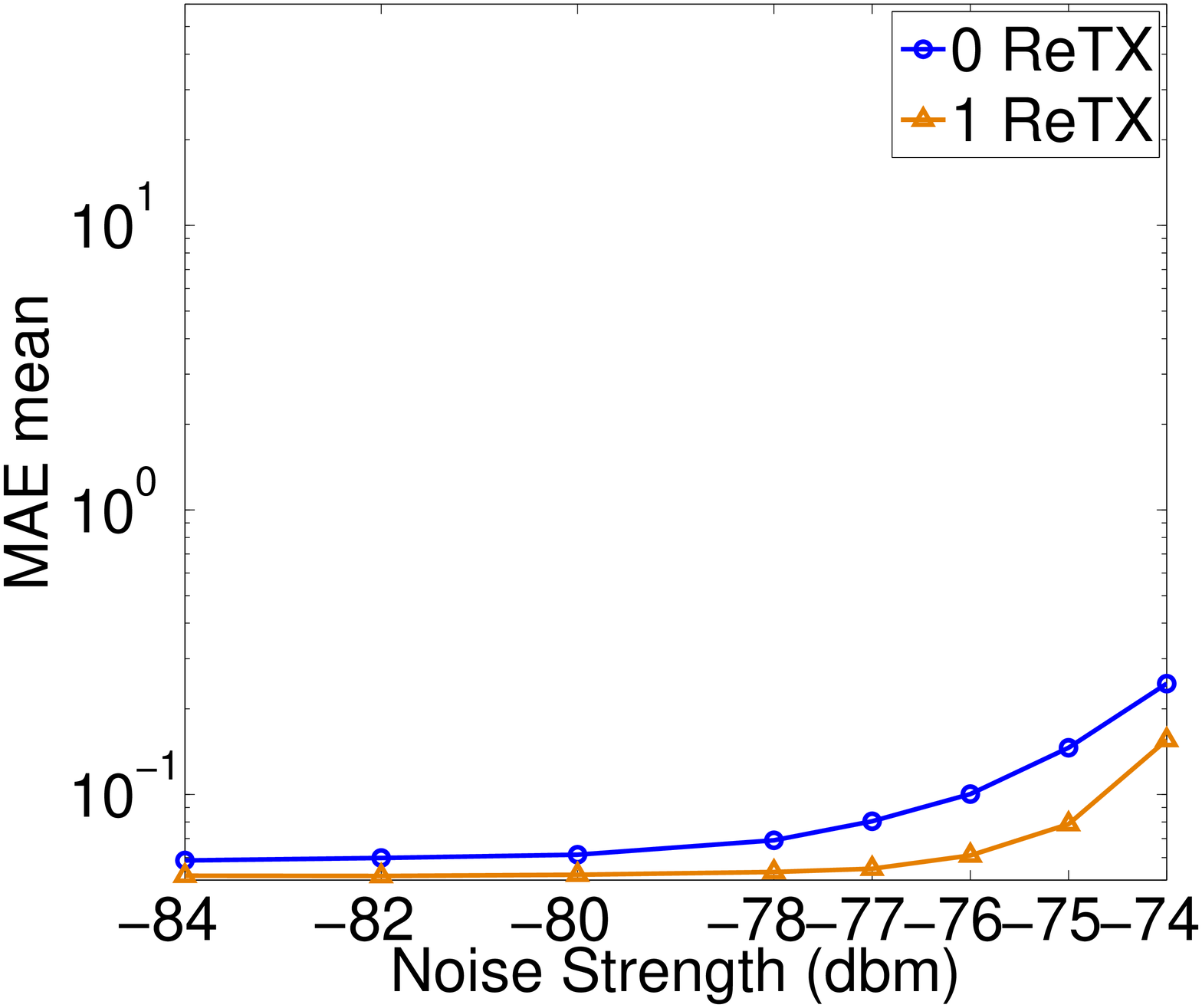}%
        \subcaption{Open-loop stable plant}
           \label{fig:rstable}%
            \end{minipage}
	\caption{%
    Mean Absolute Error (MAE) for different retransmission indices under different wireless background noise conditions, evaluated on two linear systems: one open-loop unstable (a), the other open-loop stable (b).
    Note how higher PDRs have different consequences depending on the properties of the physical plant.
  }
	\label{fig:delivery}
\end{figure}

On the other hand, higher packet delivery ratios do not immediately imply a good closed-loop performance in the physical plant, as shown in Fig.~\ref{fig:delivery}. \reply{(The settings of this set of experiments are the same with Section.~\ref{sec:results})}
Indeed, internal properties of the physical plant, such as its stability, limit and shape the effect that improvements in network communication have on the final control objective.
This is our motivation to create a new holistic controller that collects and intertwines information from both physical system and wireless network, adapting all the available parameters towards the goal of stable and safe physical executions.

In this paper, we do not aim only to find theoretical sufficient conditions to guarantee the stability of the physical plant, but we also aim to provide an implementable algorithm for the holistic controller and network manager.
\section{Holistic Controller Design}
\label{sec:controllerdesign}
%\note{Balance descriptions of stochastic properties, also do this in related works.}

In this section, we focus on three key areas to achieve our objective of designing a stable controller over a WirelessHART multi-hop network.
First, we show that in an ideal case, where the network delivers every packet with no delay, the MPC scheme in~\eqref{eq:mpc} results in asymptotically stable executions.
We achieve this goal by showing that the function $V(x)$, defined in~\eqref{eq:mpc}, is in fact a Lyapunov function.
Second, we find theoretical bounds for the worst-case evolution of the Lyapunov function $V(x)$.
Thus, if the wireless control architecture described in Section.~\ref{sec:wcs} results in values of the Lyapunov function that violate the worst-case bounds, then we must adjust the parameters in the wireless network to increase reliability.
Third, we use this adaptation principle to build a holistic algorithm that guarantees the stability of the physical plant while simultaneously reducing the latency and power usage of the wireless network.

Without loss of generality, we assume throughout this section that $\bar{x} = 0$ and $\bar{u} = 0$, as defined in~\eqref{eq:state_ref}, to simplify our notation.

\subsection{Stable Control of the Physical Plant}
\label{sec:cost_lyapunov}

Our holistic controller uses a combination of theoretical guarantees and real-time observations to decide how many transmissions must schedule each node in the wireless network.
At the core of our algorithm is the guarantee that, using an ideal network, the MPC scheme defined in Section.~\ref{sec:wcs} results in (exponentially) asymptotically stable trajectories.

We follow the strategies described in~\cite{Scokaert1998} and~\cite[Sec.~3.3]{Mayne2000} to prove the stability of our MPC scheme.
In particular, using the notation in~\eqref{eq:mpc}, given the matrices $Q$ and $R$, we compute $P$ as the unique positive definite solution of the following discrete-time algebraic Riccati equation:
\begin{equation}
  \label{eq:uc_riccati}
  \begin{aligned}
     P = A^\tp\, P\, A + Q - A^\tp\, P\, B\, \p{R + B^\tp\, P\, B}^{-1}\, B^\tp\, P\, A,
   \end{aligned}
\end{equation}
and we define:
\begin{equation}
  \label{eq:uc_lqr}
  K = - \p{R + B^\tp\, P\, B}^{-1}\, B^\tp\, P\, A.
\end{equation}

\begin{lemma}
  \label{lemma:lyap}
  Assume that ${\mathcal X}$ and ${\mathcal U}$ are polytopes, i.e., ${\mathcal X} = \set{x \mid \Gamma_x\, x \leq b_x}$ and ${\mathcal U} = \set{u \mid \Gamma_u\, u \leq b_u}$, and consider the MPC scheme in~\eqref{eq:mpc}.

  If $P$ is defined as in~\eqref{eq:uc_riccati}, $K$ is defined as in~\eqref{eq:uc_lqr}, $S = \beta\, P$ for $\beta \geq 1$, and:
  \begin{equation}
    \label{eq:xf_set}
    {\mathcal X}_f = \set{x \mid \smat{\Gamma_x\\ \Gamma_u\, K}\, x \leq \smat{b_x\\ b_u}},
  \end{equation}
  then the system in~\eqref{eq:diffeq} is asymptotically stable and $V(x)$ is a Lyapunov function.
\end{lemma}
\begin{proof}
  It is sufficient to show that we satisfy the conditions in Assumptions~A1 to~A4 in~\cite[Sec.~3.3]{Mayne2000}, where $\kappa_f(x) = K\, x$.
  Indeed, Assumptions~A1 and~A2 are trivially satisfied thanks to our definition in~\eqref{eq:xf_set}.
  Assumption~A3 follows Lemma~1 in~\cite{Scokaert1998}, and Assumption~A4 is satisfied since a simple algebraic manipulation of~\eqref{eq:uc_riccati} implies that:
  \begin{equation}
    \label{eq:pf_lemma1}
    \p{A + B\, K}^\tp\, S\, \p{A + B\, K} - S + Q + K^\tp\, R\, K =
     (1 - \beta)\, \p{Q + K^\tp\, R\, K},
  \end{equation}
  where the right-hand side is a negative semi-definite matrix, as desired.
\end{proof}
Note that $\beta$ allows us to easily adjust the transient response while maintaining the stability guarantee.
Also note that in pathological situations, the set ${\mathcal X}_f$ in~\eqref{eq:xf_set} could be empty or have no interior.
A discussion regarding those situations is beyond the scope of this paper; we refer the interested reader to~\cite[Ch.~5.2.3]{Boyd2004}.

Our holistic controller uses the value of the Lyapunov function $V(\hat{x}_t)$ to test if the wireless network has an undesired impact over the performance of the physical plant.
Our test requires calculating three parameters, $\set{\alpha_i}_{i=1}^3$, as explained below.

\begin{lemma}
  \label{lemma:alpha_1}
  Consider $V(x)$ as defined in~\eqref{eq:mpc}.
  Let $P_0$ be positive definite matrix computed recursively via the time-varying discrete-time Riccati equation:
  \begin{equation}
    P_{k-1} \!= A^\tp\, P_k\, A + Q - A^\tp\, P_k\, B\, \p{R + B^\tp\, P_k\, B}^{-1}\, B^\tp\, P_k\, A,
  \end{equation}
  with $P_N = S$, and with $\alpha_1$ as the smallest eigenvalue of $P_0$.
  Then $V(x) \geq \alpha_1\, \norm{x}^2$.
\end{lemma}
\begin{proof}
  If we relax the optimization problem by setting ${\mathcal U} = \R^m$ and ${\mathcal X} = {\mathcal X}_f = \R^n$, then we obtain a finite-horizon LQR problem.
  As explained in~\cite[Ch.~4.1]{Bertsekas2005}, the value of a finite-horizon LQR problem with initial condition $x$ is $x^\tp\, P_0\, x$; hence $V(x) \geq x^\tp\, P_0\, x \geq \alpha_1\, \norm{x}^2$.
\end{proof}
Note that if we choose $\beta = 1$, then $P_0 = P$, as defined in~\eqref{eq:uc_riccati}; thus we simplify the numerical calculation of $\alpha_1$.

\begin{lemma}
  \label{lemma:alpha_2}
  Consider $V(x)$, $A$, $Q$, and $S$ as defined in~\eqref{eq:mpc}.
  Let:
  \begin{equation}
    M = \sum_{j=0}^{N-1} \p{A^j\, Q\, A^j} + A^N\, S\, A^N,
  \end{equation}
  and let $\alpha_2$ be the largest eigenvalue of $M$.
  Then $V(x) \leq \alpha_2\, \norm{x}^2$.
\end{lemma}
\begin{proof}
  The proof follows directly after noting that $u_j = 0$ for each $j$ is a feasible input signal.
  Thus $V(x) \leq x^\tp\, M\, x \leq \alpha_2\, \norm{x}^2$.
\end{proof}

\begin{lemma}
  \label{lemma:alpha_3}
  Consider the system in~\eqref{eq:diffeq} with the closed-loop controller in~\eqref{eq:mpc}.
  Let $\alpha_3$ be the smallest eigenvalue of $Q$.
  Then $V(x_{t+1}) - V(x_t) \leq -\alpha_3\, \norm{x_t}^2$.
\end{lemma}
\begin{proof}
  Let $\set{u_{t,j}^*}_{j=t}^{t+N-1}$ be the optimal input signal associated with the value function $V(x_t)$.
  Let $\set{x_j}_{j=t}^{t+N}$ be execution resulting from applying the input $\set{u_{t,j}^*}_{j=t}^{t+N-1}$.

  Note that the input signal $\set{u_{t,t+1}^*, u_{t,t+2}^*, \dotsc, u_{t,t+N-1}^*, K\, x_{t+N}}$ belongs to the feasible set of the problem with value $V(x_{t+1})$.
  Hence:
  % \begin{multline}
  \begin{equation}
    \label{eq:pf_lemma4}
    V(x_{t+1})
    \leq \sum_{j=t+1}^{t+N-1} \p{x_j^\tp\, Q\, x_j + \p{u_{t,j}^*}^\tp\, R\, u_{t,j}^*}
    + x_{t+N}^\tp\, \pb{Q + K^\tp\, R\, K + (A+B\, K)^\tp\, S\, (A+B\, K)}\, x_{t+N},
  \end{equation}
  % \end{multline}
  and:
  \begin{equation}
    V(x_{t+1}) - V(x_t) \leq - x_t\, Q\, x_t \leq - \alpha_3\, \norm{x_t}^2,
  \end{equation}
  where we use the result in~\eqref{eq:pf_lemma1} after eliminating all the repeated terms in~\eqref{eq:pf_lemma4} and $V(x_t)$.
\end{proof}

Using the parameters $\set{\alpha_i}_{i=1}^3$, we now obtain two results that will become the test conditions to evaluate the performance loss of the physical control loop due to information loss in the wireless network.

\begin{lemma}
  \label{lemma:safety_bnd}
  Consider the system in~\eqref{eq:diffeq} with the closed-loop controller in~\eqref{eq:mpc}.
  If $V(x) \leq \alpha_1\, \gamma$, then $\norm{x}^2 \leq \gamma$.
\end{lemma}
\begin{proof}
  Suppose that $\norm{x}^2 > \gamma$, then $V(x) \geq \alpha_1\, \norm{x}^2 > \alpha_1\, \gamma$.
  The result follows, using the contrapositive of the argument above.
\end{proof}

\begin{lemma}
  \label{lemma:lyap_worst_case}
  Consider the system in~\eqref{eq:diffeq} with the closed-loop controller in~\eqref{eq:mpc}.
  Then, for each $j \in \N$:
  \begin{equation}
    V(x_{t+j}) \leq \p{1 - \frac{\alpha_3}{\alpha_2}}^j\, V(x_t).
  \end{equation}
\end{lemma}
\begin{proof}
  Using the results in Lemmas~\ref{lemma:alpha_2} and~\ref{lemma:alpha_3} we get that, for each $t \in \N$,
  $V(x_{t+1}) \leq \pb{1 - \frac{\alpha_3}{\alpha_2}}\, V(x_t)$.
  The desired result follows by induction.
\end{proof}

\subsection{$\#TX$ Adaptation}
\label{sec:retx_adapt}
As explained in Section.~\ref{sec:wcs}, among all the configuration parameters of the wireless network that we can modify, we have chosen to adapt the $\#TX$, denoted $\eta_t$, that each node in the network uses to determine the maximum number of retries used to send a packet before it is dropped.
While one can intuitively expect that more transmissions should improve the physical control loop performance, they come associated with longer delays and shorter battery lifetimes for the nodes.
Moreover, it is not immediately clear how to quantify the impact that more transmissions have in the control loop, as shown in Fig.~\ref{fig:delivery}.

The value of the Lyapunov function $V(\hat{x}_t)$ at each $t$ \reply{for each control loop}, and the bounds in the lemmas above, give us a constructive testing mechanism to evaluate the impact that the loss of information in the wireless network has on the physical plant.
Suppose that a setpoint $\bar{x}$ has been computed as in~\eqref{eq:state_ref}, and a desired maximum deviation from that setpoint, $\gamma$, has been defined, i.e. the goal is to maintain $\norm{x_t - \bar{x}}^2 \leq \gamma$ for each $t \in \N$.
To build a $\#TX$ adaptation algorithm based on the analytical results above, we must first establish a set of principles that our algorithm must satisfy:
\begin{itemize}
\item If $V(x_t) \leq \alpha_1\, \gamma$, then $\eta_{t+1} \leq \eta_t$, i.e.\ $\eta_t$ will not increase, since the physical plant is within acceptable bounds, as shown in Lemma~\ref{lemma:safety_bnd}.
\item Given $\lambda \in \p{0,1}$, if $V(x_t) \geq \lambda\, \alpha_1\, \gamma$, then $\eta_{t+1} \geq \eta_t$, i.e.\ $\eta_t$ will not decrease, since the physical plant might get closer to the safety bound in Lemma~\ref{lemma:safety_bnd}.
\item If $V(x_{t+j}) \leq \pb{1 - \frac{\alpha_3}{\alpha_2}}^j\, V(x_t)$, then $\eta_{t+j} \leq \eta_t$, i.e.\ $\eta_t$ will not increase, since the physical is evolving towards its equilibrium point within expected bounds, as shown in Lemma~\ref{lemma:lyap_worst_case}.
\item If the current PDR, denoted $\rho_t$, is below a threshold, say $\rho_t < \rho_{\min}$, then $\eta_{t+1} \geq \eta_t$, i.e.\ $\eta_t$ will not decrease, since the network must maintain a minimum connectivity level.
\end{itemize}
The parameter $\lambda$ is used to create a dead-band between increases and decreases of the $\#TX$.
Indeed, if $V(x_t) \in \sp{\lambda\, \alpha_1\, \gamma, \alpha_1\, \gamma}$ then $\eta_t$ remains constant.

% \begin{figure}[t]
\SetKwFor{Loop}{Loop}{}{EndLoop}
\begin{algorithm}[t]
  \SetAlgoNoLine
  \KwIn{$t,\tau_1,\tau_2 \in \N$, $t_0 = t$, $\lambda \in (0,1)$, $\rho_{\min} \in [0,1]$, $\eta_{\max} \in \N$, $\delta = -1$, and an initial $\#TX$ $\eta_t \in \set{1,\dotsc,\eta_{\max}}$.}
  \KwOut{$\#TX$ $\eta_t$}
  \Loop{}
  {
    Evaluate $V(\hat{x}_t)$ as defined in~\eqref{eq:mpc}\;
    Measure the PDR $\rho_t$\;
    \uIf{$V(\hat{x}_t) < \lambda\, \alpha_1\, \gamma$ \textbf{and} $\eta_t > 0$ \textbf{and} $\rho_t \geq \rho_{\min}$}
    {
      \If{$\delta \neq 0$}
      {
        $t_0 \gets t$\;
        $\delta \gets 0$\;
      }
      \If{$t - t_0 > \tau_2$}
      {
        $t_0 \gets t$\;
        $\eta_{t+j} \gets \eta_t - 1$ for each $j \leq \tau_1$\;
        $t \gets t + \tau_1$;
      }
    }
    \uElseIf{$V(\hat{x}_t) > \alpha_1\, \gamma$ \textbf{and} $\eta_t < \eta_{\max}$}
    {
      \If{$\delta \neq 1$}
      {
        $t_0 \gets t$\;
      }
      \If{$\delta \neq 1$ \textbf{or} $V(\hat{x}_t) > \pb{1 - \frac{\alpha_3}{\alpha_2}}^{t-t_0}\, V(\hat{x}_{t_0})$}
      {
        $\delta \gets 1$\;
        $\eta_{t+j} \gets \eta_t + 1$ for each $j \leq \tau_1$\;
        $t \gets t + \tau_1$\;
      }
    }
    \Else
    {
      $\eta_{t+1} \gets \eta_t$\;
      $t \gets t + 1$\;
    }
  }
  \caption{$\#TX$ adaptation algorithm \reply{for each control loop}.}
  \label{alg:holistic}
\end{algorithm}

Our $\#TX$ adaptation algorithm \reply{for each control loop} is described in detail in Alg.~\ref{alg:holistic}.
\reply{Given a control loop,} we measure the current packet delivery ratio $\rho_t$ on each iteration, and we compute $V(\hat{x}_t)$.
With these two values, we decide whether we decrease, increase, or maintain the $\#TX$.
The variable $\delta$ is used as an internal state to determine whether the last $\#TX$ change was an increase ($\delta = 1$) or a decrease ($\delta = 0$).
The parameter $\lambda \in (0,1)$ determines the width of the dead-band for the Lyapunov function where the $\#TX$ is left unchanged.
The parameter $\tau_1 \in \N$ is the minimum number of iterations it takes the wireless network to propagate the new transmission schedule, as explained in Section.~\ref{sec:wsn}.
The parameter $\tau_2 \in \N$ is used to slow down the $\#TX$ decreases, since those might eventually result in violations of the safety bound in Lemma~\ref{lemma:safety_bnd}.
Finally, the parameters $\eta_{\max}$ and $\rho_{\min}$ are chosen such that all routes can be scheduled and minimum network control information is still delivered.

\reply{%
Our control strategy cannot mathematically guarantee the closed-loop system stability of the WSAN unless extra assumptions are considered, which is a common approach in the literature~\cite{hadjicostis2002feedback, ljevsnjanin2014packetized,amin2009safe}.
%~\note{cite two or three papers related to WSAN with strong assumptions, such as those from KTH or Quevedo}.
The applicability and appropriateness of these assumptions depend on the particular properties of the industrial plant at hand, thus we avoid imposing a particular framework in this paper.
Instead, our algorithm takes a best-effort approach towards balancing closed-loop performance and network load, which is a practical heuristic in real-world scenarios.
The stability of our algorithm hinges on the relation between dropped actuation packets and their impact on the empirical Lyapunov function at each iteration, which has been studied in the past~\cite{Mazo2008,Araujo2011}.
}

% \note{About to delete.}

% The following theorem is a direct consequence of Lemmas~\ref{lemma:safety_bnd} and~\ref{lemma:lyap_worst_case}, and the algorithm in Alg.~\ref{alg:holistic}, thus we omit the proof.
% \begin{theorem}
%   Suppose that there exists a $\#TX$ $\eta^* \in \set{1,\dotsc,\eta_{\max}}$ such that, when applied to the wireless network schedule, the bound in Lemma~\ref{lemma:lyap_worst_case} is satisfied.
%   Then the algorithm in Alg.~\ref{alg:holistic} results in a stable trajectory. %, and there exists $\lambda \in (0,1)$ such that $\norm{x_t - \bar{x}}^2 > \gamma$ for only a finite number of iterations.
% \end{theorem}
 % algorithm not ready --Yehan
\section{Network Reconfiguration}
\label{sec:wsn}

In this section, we introduce a run-time reconfiguration protocol for the WSAN. \reply {Our previous research has demonstrated that a wireless control system can have different levels of resilience to packet loss for sensing and actuation.}
Motivated by the asymmetric routing idea in~\cite{li2016wireless},
We develop an asymmetric scheduling approach in which the number of packet transmissions of the sensing and actuation phase can be configured independently.
Considering that sensing data are less vulnerable against packet loss because of the state observer, we do not allocate retransmissions for sensing packets.
However, we allow a holistic controller to adaptively adjust the number of transmissions for actuation packets \reply{of each control loop} based on the physical and network conditions. \reply{This need for adjustment stems from the fact that the control performance is more sensitive to packet loss from the controller to the actuators despite the buffered control inputs.}

We next present a run-time transmission adaptation protocol.
In our design, the network manager generates a schedule that allocates enough slots to accommodate the maximum $\#TX$ over each link defined by the TX adaptation algorithm: if a transmission is scheduled at time slot $x$, then the scheduler will reserve the next $n-1$ consecutive slots for its retransmissions, where $n$ is the maximum $\#TX$ per hop.
Each schedule entry is represented by a tuple [slot\_offset, channel, sender, receiver, flowID, \#TX(flowID)].
A transmission schedule is called a superframe, which repeats itself during runtime. The slot\_offset is the relative time slot number in a superframe. The flowID specifies the flow a transmission belongs to, and \#TX(flowID) indicates the current $\#TX$ of this flow, enabling the protocol to configure $\#TX$ of each flow independently.
In Table~\ref{table:piggyback}, $1 TX$ to $3 TX$ presents a TDMA-schedule for $2$ flows $F1$ and $F2$ that deliver data through nodes A $\rightarrow$ B $\rightarrow$ C and A $\rightarrow$ B $\rightarrow$ D, respectively, when \#TX(F1) and \#TX(F2) vary from $1$ to $3$ transmissions. The superframe has a length of 12 time slots. Note that other routing and scheduling algorithms exist that optimize network resource usage to enhance network scalability, but these are not within the scope of this work.

We adopt a {\em piggyback} mechanism to disseminate a newly computed $\#TX$ for a certain control loop (flow) generated by the holistic controller. A network manager, which is co-located with a holistic controller, can incorporate the updated $\#TX$ and the $\#Flow$ into all the periodic actuation packets in this control loop.
Hence, all nodes along the actuation routes of certain control loops can receive this update.
This piggyback mechanism helps reduce communication cost by utilizing existing periodic communication.
Once a node receives a $\#TX$ switch command, it will apply a new $\#TX$ at the beginning of next superframe.
However, if a node fails to receive the command due to packet loss, it will continue to use the current $\#TX$ until any actuation packet is received.
Therefore, it is possible that, at the same time, different nodes along the route of a flow may use different $\#TX$.

Nevertheless, it is still possible for nodes to eventually receive the update since they are always scheduled to communicate at the slot allocated for their first transmission attempts over a link.
For example, in Table~\ref{table:piggyback}, the transmissions colored red represent the current schedule for flow F1 (A $\rightarrow$ B $\rightarrow$ C) and F2 (A $\rightarrow$ B $\rightarrow$ D), when a controller issues a command to update \#TX(F1) from 1 to 2, and \#TX(F2) $= 3$. In this example, the update reaches node A and B at time slot 1, but fails to arrive at C at time slot 4 due to packet loss. Hence, A and B will switch to 2TX, while C remains to use 1TX. Although B and C use different $\#TX$, it is still possible for C to receive actuation and mode switch commands from B in the following superframe since B and C will always communicate at slot 4 regardless of $\#TX$. The $\#TX$ of $F2$ is kept unaltered during the process since there is no $\#TX$ update for F2.

\begin{table}[tp]
\small
  \caption{Piggyback Algorithm Superframe Examples}
  \label{table:piggyback}
  \resizebox{\columnwidth}{!}{%
    \begin{tabular}{c c c c c c c c c c c c c}
      \hline
      $\#TX$ & Slot 1 & Slot 2 & Slot 3 & Slot 4 & Slot 5 & Slot 6&Slot 7&Slot 8&Slot 9&Slot 10&Slot 11&Slot12\\
      \hline
      $1$ & A$\rightarrow$B & & & \color{red}{B$\rightarrow$C} & & & A$\rightarrow$B&&& B$\rightarrow$D&&\\
      $2$ & \color{red}{A$\rightarrow$B} & \color{red}{A$\rightarrow$B} & & B$\rightarrow$C & B$\rightarrow$C &&A$\rightarrow$B&A$\rightarrow$B&&B$\rightarrow$D&B$\rightarrow$D \\
      $3$ & A$\rightarrow$B & A$\rightarrow$B & A$\rightarrow$B & B$\rightarrow$C& B$\rightarrow$C & B$\rightarrow$C&\color{red}{A$\rightarrow$B}&\color{red}{A$\rightarrow$B}&\color{red}{A$\rightarrow$B}&\color{red}{B$\rightarrow$D}&\color{red}{B$\rightarrow$D}&\color{red}{B$\rightarrow$D} \\
      \hline
    \end{tabular}%
  }
\end{table}

It is worth noting that our paper mainly discusses WSANs that are revivable under moderate cyber and physical attacks.
This is why we set $\rho_{\min}$ to guarantee relatively high PDRs of wireless networks in Section.~\ref{sec:retx_adapt}.
Under extreme conditions, a larger portion of piggyback packets may be lost.
In this case, sending commands to switch modes by broadcasting or flooding might be a better solution.
\section{Case Study}
\label{sec:results}

In this section, we introduce a systematic case study for our holistic WSAN controller.
On the physical side, we use two 5-state linear time-invariant plants that share the same WSAN as representative of an industrial process systems. \reply{One of the plants is open-loop unstable, and the other is open-loop stable.}
On the cyber side, we simulate a 16-node WSAN using the WCPS simulator~\cite{wcps_web}, seeded with real-world wireless traces as explained in~\cite{li_iccps13, li2015incorporating, li2016wireless}.
WCPS works by fully integrating Simulink~\cite{Simulink} and TOSSIM~\cite{Lee07:tossim-noise}.
Besides studying the behavior of the algorithms presented in Sections.~\ref{sec:controllerdesign} and~\ref{sec:wsn}, we also test the performance of our system under both cyber and physical disturbances. \reply{ We will present the case study of the open-loop unstable plant in the first three sections. Then, we will include an open-loop stable plant that shares WSAN with that open-loop unstable plant to form a multi-loop simulation in Section.~\ref{sec:multi}.}

\subsection{Experimental Setting}
\label{sec:setting}
The plant is defined in~\eqref{eq:diffeq}, with the following parameters:
\begin{equation}
  %\begin{gathered}
    A =
    \begin{bmatrix}
      0.717 & -1.367 & -0.218 & -0.867 & -0.899 \\
      0.078 & 0.209 & -0.105 & -0.511 & -0.466\\
      0.122 & 0.891 & 1.305 & 0.511 & 0.666\\
      -0.243 & -1.383 & -0.610 & -0.023 & -0.932\\
      0.122  & 0.871 & 0.165& 0.712 & 1.466
    \end{bmatrix},\quad
    B =
    \begin{bmatrix}
      0.083\\ 0.056\\ -0.056\\ 0.111\\ -0.056
    \end{bmatrix},\quad\text{and}\quad
    C =
    \begin{bmatrix}
      0\\ 1\\ 1\\ 0\\ 0
    \end{bmatrix}^\tp.
  % \end{gathered}
\end{equation}
Note that the set of eigenvalues of $A$ is equal to $\set{0.413, 0.563, 0.624, 1.068, 1.006}$.
Since there are two eigenvalues outside the unit circle, the plant is open-loop unstable.
% According to Lemma 1 in Sec.~\ref{sec:cost_lyapunov}, we can easily find MPC parameters that satisfy Lyapunov function is Lyapunov function without network communication.

% \subsubsection{MPC Parameters}
The parameters of the MPC scheme in~\eqref{eq:mpc} are chosen as follows:
\begin{equation}
  \begin{gathered}
    Q =
    \begin{bmatrix}
      1 & 0 & 0 & 0 & 0 \\
      0 & 1 & 0 & 0 & 0\\
      0 & 0 & 1 & 0 & 0\\
      0 & 0 & 0 & 5 & 0\\
      0 & 0 & 0 & 0 & 1
    \end{bmatrix},\quad
    S =
    \begin{bmatrix}
      439 & 518.37 & 942.57 & 220.13 & 588.16 \\
      518.37& 633.85 & 1124.8 & 266.03 & 708.54\\
      942.57 & 1124.8 & 2045.3 & 483.1 & 1274.2\\
      220.13 & 266.03 & 483.1 & 129.01& 309.01\\
      588.16 & 708.54 & 1274.2 & 309.01 & 809.37
    \end{bmatrix},\quad
    \bar{x} =
    \begin{bmatrix}
      0.289\\ 1.735\\ 0.578\\ -1.157\\ -1.735
    \end{bmatrix},\\
    K =
    \begin{bmatrix}
      -20.49 & -20.48 & -42.66 & -11.24 & -26.38
    \end{bmatrix}.
  \end{gathered}
\end{equation}
Also, $\bar{u} = 0.2$, $R = 0.08$, $N = 50$, ${\mathcal X} = \R^5$, and ${\mathcal U} = \set{\abs{u} \leq 40}$, where $S$, $K$, and ${\mathcal X}_f$ are computed as described in Lemma~\ref{lemma:lyap} with $\beta = 1.1$.

The parameters of the Alg.~\ref{alg:holistic} are
$\lambda = 0.1$, $\tau_1 = 5$, $\tau_2 = 500$, $\gamma = 16$, $\alpha_1 = 1.977$, $\alpha_2 = 8.223 \cdot 10^6$, and $\alpha_3 = 1$,
where $\set{\alpha_i}_{i=1}^3$ are computed as described in Lemmas~\ref{lemma:alpha_1} to~\ref{lemma:alpha_3}.

\reply{We use WCPS~\cite{wcps_web}, which was developed by the Cyber-Physical Systems Laboratory at Washington University, as a platform for holistic wireless control system simulations.
WCPS employs a federated architecture that integrates Simulink~\cite{Simulink} for simulating the physical system dynamics and controllers, and TOSSIM~\cite{Lee07:tossim-noise} for simulating WSANs. Both Simulink and TOSSIM are among the leading simulators in the control and networking communities, respectively, but with little interaction between them in the past.
Furthermore, the WirelessHART network protocol stack is implemented as part of WCPS 2.0~\cite{li2015incorporating}, including protocols at the routing (Source and Graph Routing) and MAC layers.
To the best of our knowledge, WCPS is the first simulator that supports all these WirelessHART features with a realistic wireless link model.

In this paper, we have incorporated the new $\#TX$ reconfiguration mechanism presented in Section.~\ref{sec:wsn} into WCPS.
Additionally, we have included multi-rate functionality into the simulation.
In reality, the industrial plant models are mostly continuous or running at very high frequencies. However, the wireless  communication and controller execute at a relatively low frequency because of the network latency and computational latency.
Therefore, we incorporate multi-rate functionality in WCPS, which simulates plant, wireless network, and controller in multiple rates.
Additionally, we have considered the latency of each module in the wireless control loop.
Since sensing measurements and control commands are sent via WSAN periodically, time-driven scheduling is adopted, as is shown in Fig.~\ref{fig:sche}. We determine the length of each event based on its worst-case execution time.}
%WCPS has already proven to be an valuable tool for CPS research on wireless control systems. Using WCPS, we have developed case studies that simulate wireless control systems for civil infrastructure [67] and process plants [65,66], allowing us to motivate our cyber-physical co-design of wireless communication protocols and control algorithms ideas.}
% \subsubsection{Wireless Settings}
%Network manager will change ReTX number of the downstreams in WSAN in the beginning of the third superframe after the superframe that it receives ReTX switch command from holistic controller.}

We simulate a wireless network consisting of 16~nodes, where each simulation is based upon data traces collected from our WSAN testbed at Washington University.
The collected information includes \emph{received signal strength indicator} (RSSI) and electromagnetic noise, which are used as inputs for the wireless link model in WCPS.
The WSAN in this study consists of 6~sensing flows and 2~actuation flows.
%but we only use one of each to control our plant.
A sensing flow delivers sensing data from a sensor node to the controller, while an actuation flow delivers control commands to an actuator.
\reply{We only use two pairs of sensing and actuation flows for two control loop.} \reply{The extra sensing flows are reserved for redundant measurements, as it is commonly designed in industrial scenarios.}
The maximum distance from a sensor to an actuator is 4~hops.

\reply{Since WirelessHART employs the Time Slotted Channel Hopping technology (TSCH) MAC, the superframe length of this WSAN is $140\, \text{ms}$, i.e.\ fixed 14~time slots. The WSAN sensing delay ($T_S$) and actuation delay ($T_A$) are set to $60$ ms and $80$ ms, respectively. The worst-case execution times of the controller and the state observer are  $30$ ms and $0.2$ ms among $10000$ operations, respectively. Hence, we set the computation delay of controller ($T_C$) to $30$ ms, and Kalman Filter delay ($T_{KF}$) to $10$ ms, \reply{because the granularity of our simulation is $10\, \text{ms}$.}}
\reply{Therefore, we set the frequencies of the WSAN, the MPC controller, and the KF state observer in our simulations to $5\, \text{Hz}$. And the frequencies of the plants are set to $100\, \text{Hz}$. Note that the control command executed at time $T_k$ is based on the sensor measurements at $T_{k-1}$, as is shown in Fig.~\ref{fig:sche}. Since we are using MPC, at time $T_k$, we use $\setn{u_{j,k}^*}_{j=k-N}^{k-1}$. In this way, we mitigate the delay  of all the modules and the effects of the actuation packet loss in the wireless control loop.}
%
%, which filters out the impact of the control frequency while allowing us to focus on studying the impact of retransmissions.

\reply{Because of the asymmetric nature of sensing and actuation sides}~\cite{li2016wireless}, we adopt an \emph{asymmetric} \reply{scheduling} strategy.
For sensing flows, we do not provide any packet retransmission since the state observer mitigates the impact of packet loss, as explained in Section.~\ref{sec:wcs}.
For actuation flows, we set the maximum $\#TX$ to $\eta_{\max} = 4$ (As explained in Section.~\ref{sec:WNM}, there is a diminishing return of PDR improvements as $\#TX$ increases. Therefore, we set $\eta_{\max} = 4$ to efficiently improve PDR at reasonable cost.), and the minimum packet delivery ratio to $\rho_{\min} = 0.7$.
%We refer interested users to~\cite{li2016wireless} for more details of the asymmetric idea between sensing and actuation.

%The superframe length of 0 ReTX in sensing side and 3 ReTX in actuation side is $140ms$. Thus the sampling time of wireless control system is set as $5Hz$.
%We treat the transient time of ReTX counts changes in WSAN as a fixed delay of two superframes in our current experimental settings.
%Network manager will change ReTX number of the downstreams in WSAN in the beginning of the third superframe after the superframe that it receives ReTX switch command from holistic controller.
%\subsubsection{Compare with baseline}
%Noise level is not change.

%%%%%%% HG: I'm here

\begin{figure}[tp]
  \centering
  \includegraphics[width=0.8\columnwidth]{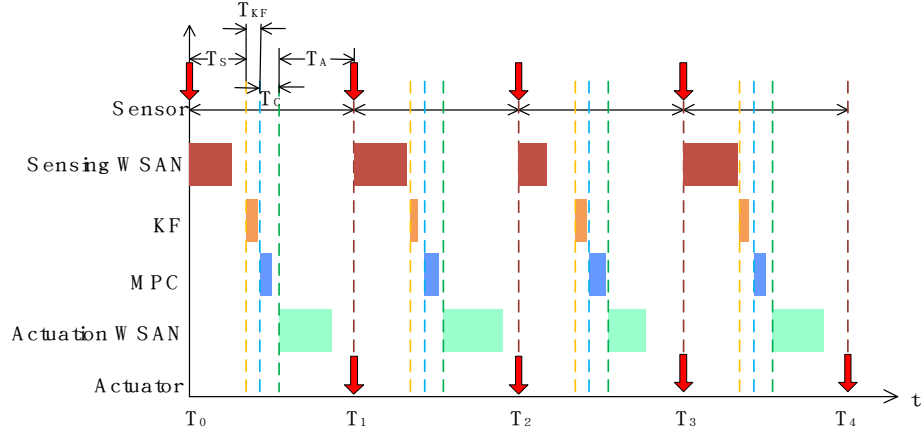}
  \caption{
    Time-driven scheduling of wireless control system
  }
  \label{fig:sche}
\end{figure}

\subsection{Simulation Under Wireless Interference}
\label{sec:eval_cyber}

\begin{figure}[tp]
  \centering
  \includegraphics[width=0.9\columnwidth]{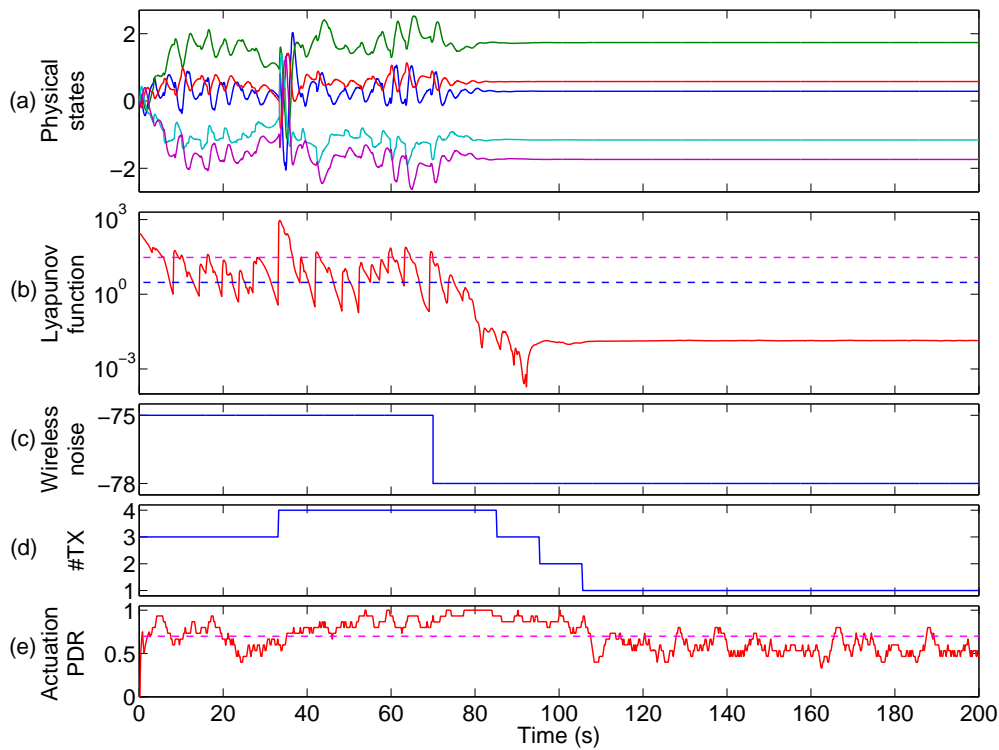}
  \caption{%
    Holistic controller simulation under wireless interference.
    The wireless background noise is higher for the first $70\, \text{s}$, resulting in a  $\#TX$ adaptation.
  }
  \label{fig:cymc}
\end{figure}

We consider a disturbance on the wireless network caused by an increase in background noise on all channels, which is common when the wireless sensor network is under wide-band continuous jamming attack.
In this case, channel hopping and channel blacklisting functionality of WirelessHART standard might fail to resist such attack.
Our simulation considers an increased value for the background noise for the first $70\, \text{s}$ at $-75\, \text{dBm}$, and a reduction to a standard value for the rest of the simulation at  $-78$dbm, as shown in Fig.~\ref{fig:cymc}c.
Note that the physical plant is unstable and in a transient state for the first $70\, \text{s}$; hence a low actuation packet delivery ratio should likely lead to diverging trajectories.

In this simulation (Fig.~\ref{fig:cymc}), there are two interesting observations.
First, the \#TX is increased twice, first at $t = 0\, \text{s}$, and then at $t \approx 33.2\, \text{s}$.
The first increase is due to a violation of the safety bound condition in Lemma~\ref{lemma:safety_bnd}, while the second is due to a violation of the worst-case Lyapunov evolution in Lemma~\ref{lemma:lyap_worst_case}.
By time $t \approx 80\, \text{s}$, the physical plant is well below the safety bound, implying that Alg.~\ref{alg:holistic} successfully adapted to the higher wireless background noise.
Second, once the background noise is reduced, our algorithm slowly decreases the $\#TX$, and finally stabilizing at a point where the physical plant is stable even under a significant number of actuation packet drops.
Note that once the physical plant reaches a steady state, the optimal input sent through the network is almost constant, with $u_{t,j}^* \approx \bar{u}$ for each $j$.
In this case, the actuation buffer, explained in Section.~\ref{sec:wcs} and Fig.~\ref{fig:wsn}, almost completely mitigates the information lost in the wireless network due to lower $\#TX$.

\begin{figure}[tp]
	\centering
  \includegraphics[width=0.75\columnwidth]{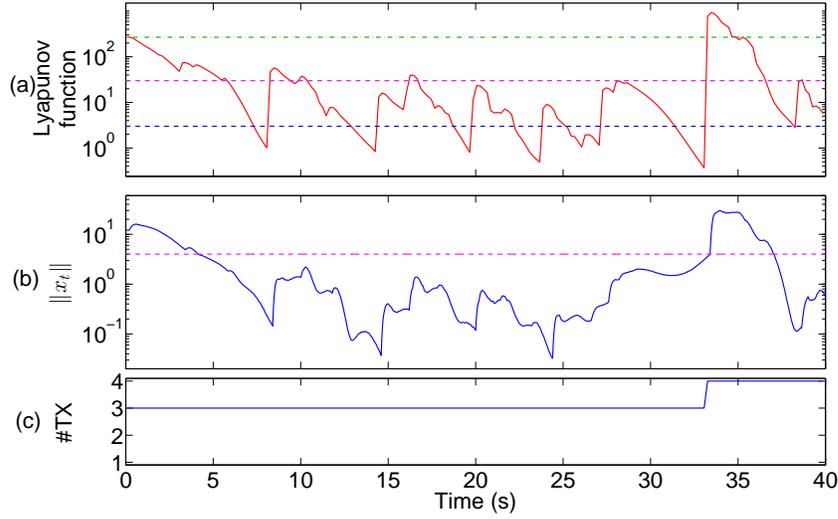}%
  \caption{
    Zoom in of the simulation in Fig.~\ref{fig:cymc} for the time interval $[0,40]\, \text{s}$.
  }
  \label{fig:detail}
\end{figure}

Fig~\ref{fig:detail} is a zoom-in view of the first 40 seconds of the simulation in Fig.~\ref{fig:cymc}.
Fig.~\ref{fig:detail}a and Fig.~\ref{fig:detail}b validate our result in Lemma~\ref{lemma:safety_bnd} of using the Lyapunov function as a simple and practical test for safety.
Also, the top green dashed line in Fig.~\ref{fig:detail}a is the bound induced by the worst-case evolution of the Lyapunov function, as explained in Lemma~\ref{lemma:lyap_worst_case}, which results in an increase in the  $\#TX$ when violated.

We compare \reply{cyber and physical properties of wireless control system} when using different algorithms to adapt the network configuration.
We compare our algorithm, HC (Alg.~\ref{alg:holistic}), against: (1) an adaptation algorithm based only on two PDR thresholds, namly, $\#TX$ increase threshold ($80\%$) and $\#TX$ decrease threshold ($90\%$). We regard this algorithm as pure network solution, denoted PN; (2) constant $\#TX$, denoted 2, 3, and 4. \reply{If an actuation packet adopts 1 TX, the average packet delivery ratio is around $23\%$ under noise level of $-75$ dbm, and around $53\%$ even under a noise level of $-78$ dbm, which will damage system performances. Thus, we set at least 1 retransmission to actuation packets.}
Each algorithm is simulated 40 times, and unstable executions are discarded to avoid distorting the average computations.
It is worth noting that fixed $2 TX$s results in 10~unstable simulations, and both PN and fixed 3, 4 $TX$s result in 1~unstable simulation, while HC stabilizes all the executions.

\begin{figure}[tp]
\centering
\begin{minipage}[b]{0.4\textwidth}
\centering
  \includegraphics[width=0.83\columnwidth]{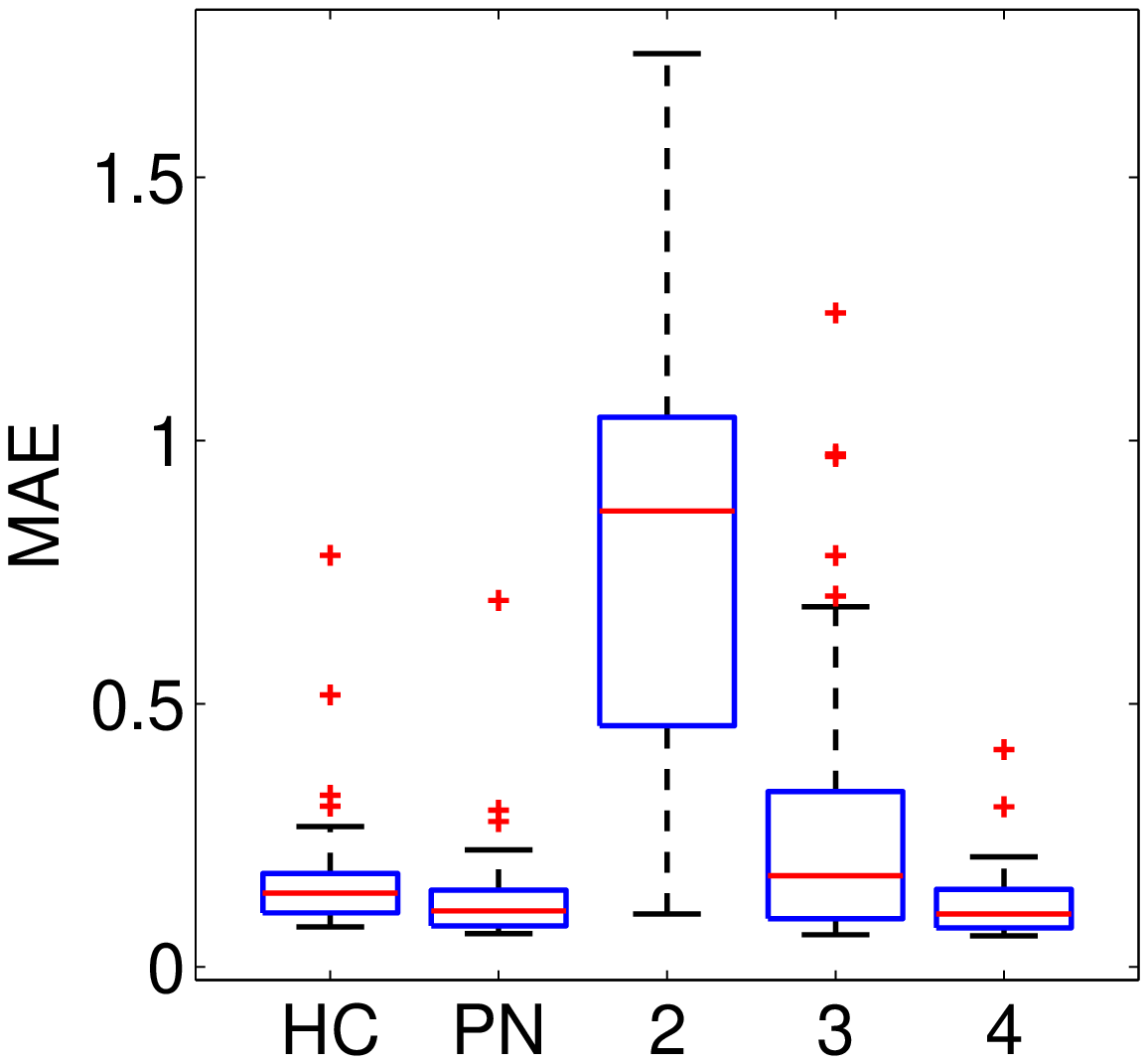}
  \subcaption{Mean Absolute Error (MAE)}
 \label{fig:cyattack1}
 \end{minipage}
\begin{minipage}[b]{0.4\textwidth}
\centering
  \includegraphics[width=0.78\columnwidth]{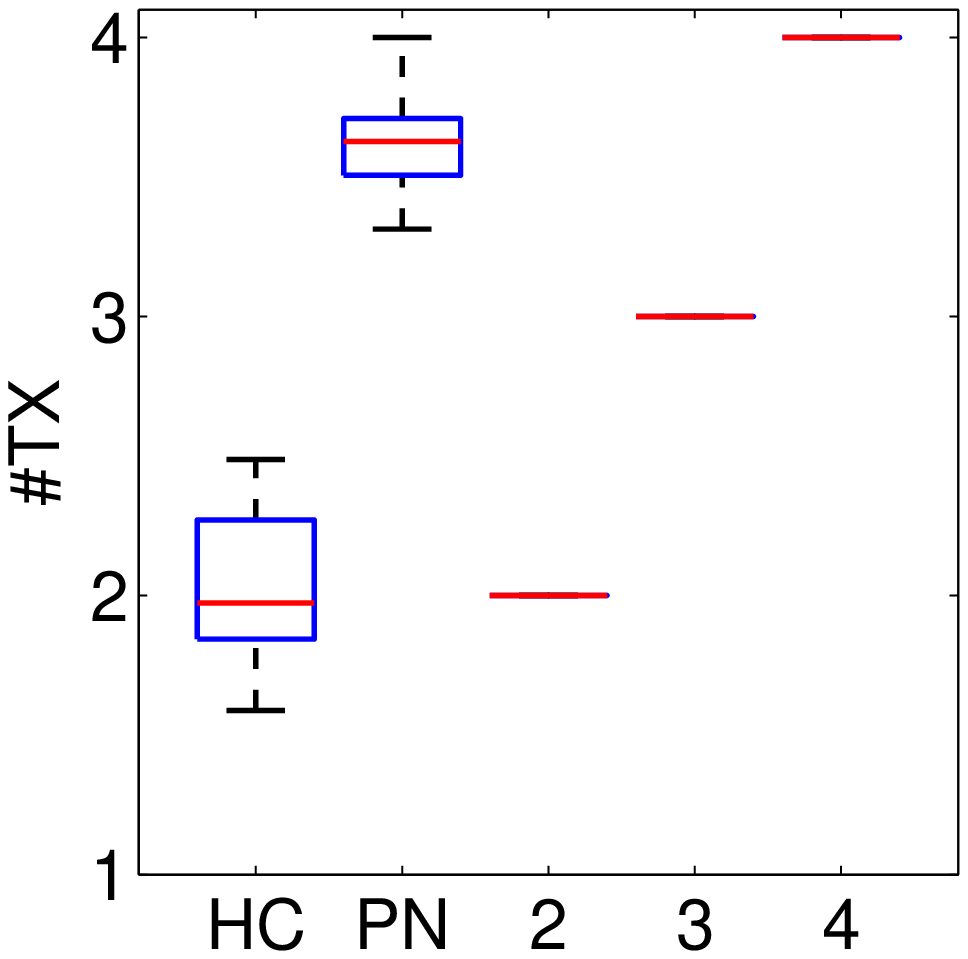}
    \subcaption{Average number of scheduled TXs}
 \label{fig:cyattack2}
 \end{minipage}
\begin{minipage}[b]{0.4\textwidth}
\centering
  \includegraphics[width=0.85\columnwidth]{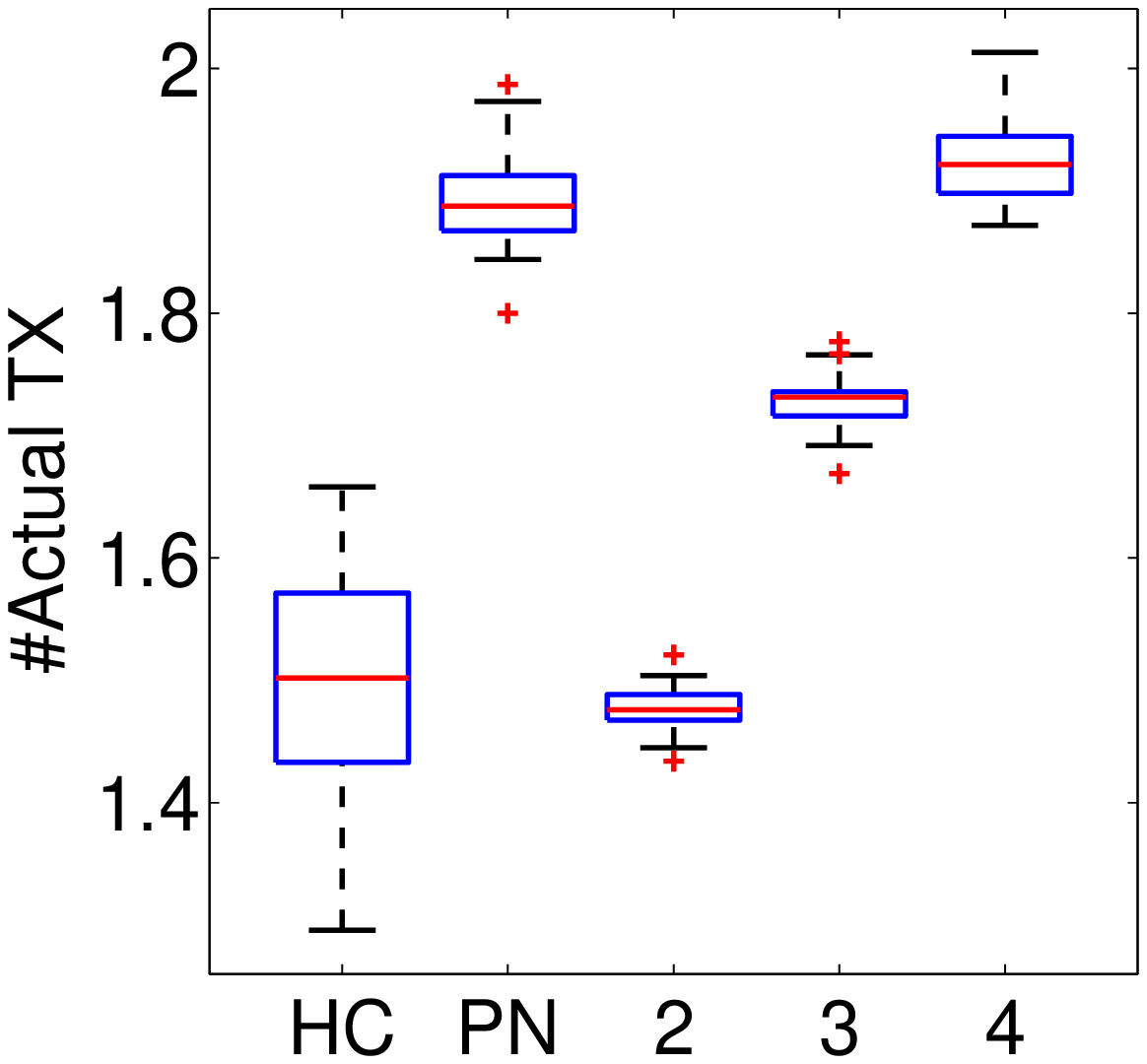}
    \subcaption{Average number of actual TXs per packet}
 \label{fig:cyattack3}
 \end{minipage}
\begin{minipage}[b]{0.4\textwidth}
\centering
   \includegraphics[width=0.85\columnwidth]{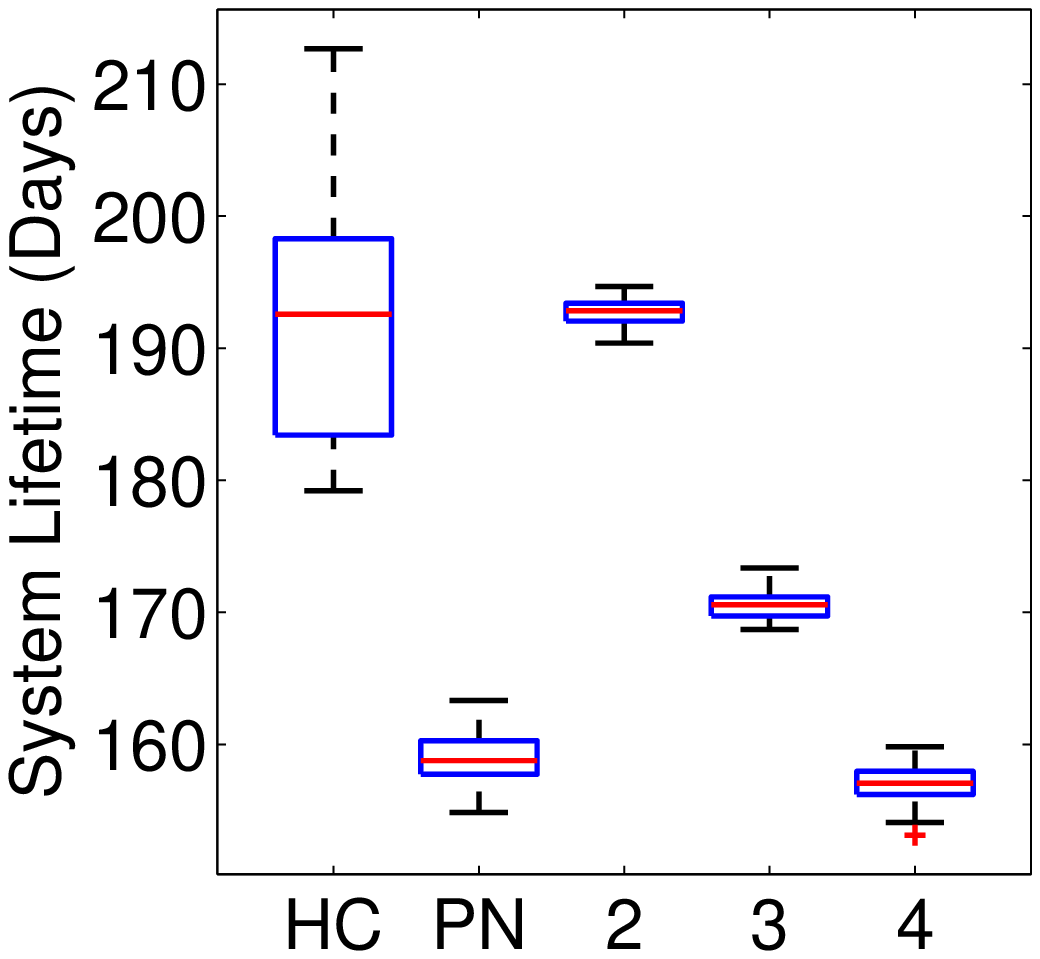}
     \subcaption{Battery life}
 \label{fig:cyattack4}
 \end{minipage}
  \caption{
    (a) MAE, (b) average $\#TX$, (c) average $\#Actual\ TX$, and (d) battery life for different $\#TX$ adaptation algorithms under wireless interference.
    The algorithms are our holistic controller (HC), pure network adaptation (PN), and constant transmission number equal to 2, 3, and 4, respectively.
  }
  \label{fig:cyattack}
\end{figure}

Fig.~\ref{fig:cyattack} compares (a) the Mean Absolute Error (MAE)  of the physical states, (b) the number of scheduled transmissions ($\#TX$ for short in the rest of the paper), (c) the average number of actual transmissions per actuation packet ($\#Actual\ TX$ for short in the rest of the paper), and (d) the system lifetime with various wireless network configuration methods.
\reply{In reality, $\#Actual\ TX$ is often less than scheduled $\#TX$. For one reason, if the sender gets the Acknowledgement from the receiver at the first try, it will not send an extra time since the packet has already been received. For another reason, if the packet is lost in previous hops, the sender has no packet to send. The $\#Actual\ TX$ is one determinant of battery life.}
\reply{We assume wireless motes use AA batteries, the capacity of which is $8640$ J. We define the system life time as the battery life of the most consuming node in the network, and calculate battery life based on the method in~\cite{wu2015conflict}.}
According to Fig.~\ref{fig:cyattack1} and Fig.~\ref{fig:cyattack2}, while our HC algorithm has a comparable MAE to PN and 4 $TX$s, its average scheduled $\#TX$ is around $2$.
\reply{Fig.~\ref{fig:act_send_cy} shows the ratio for $\#Actual\ TX$ with various wireless network configuration methods.
HC has the highest ratio of $0$ and $1$ actual $TX$ per packet, since the system performances are sometimes acceptable even though PDR is not high enough.
At the same time, HC can also adjust the $\#Actual\ TX$ to $4$ to guarantee control performances when needed.}
\reply{Fig.~\ref{fig:cyattack3} and Fig.~\ref{fig:cyattack4} shows that the $\#Actual\ TX$ and the system lifetime of our HC algorithm are also similar to 2$TX$. Furthermore, HC extends the system lifetime for more than one month compared with PN.}

\begin{figure}[tp]
\centering
  \includegraphics[width=0.5\columnwidth]{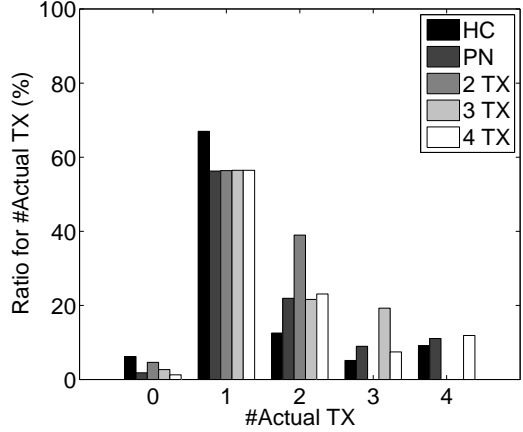}
  \caption{%
Ratio for $\#Actual\ TX$ under wireless interference
  }
  \label{fig:act_send_cy}
\end{figure}

\begin{figure}[tp]
\centering
\begin{minipage}[b]{0.45\textwidth}
\centering
  \includegraphics[width=1\columnwidth]{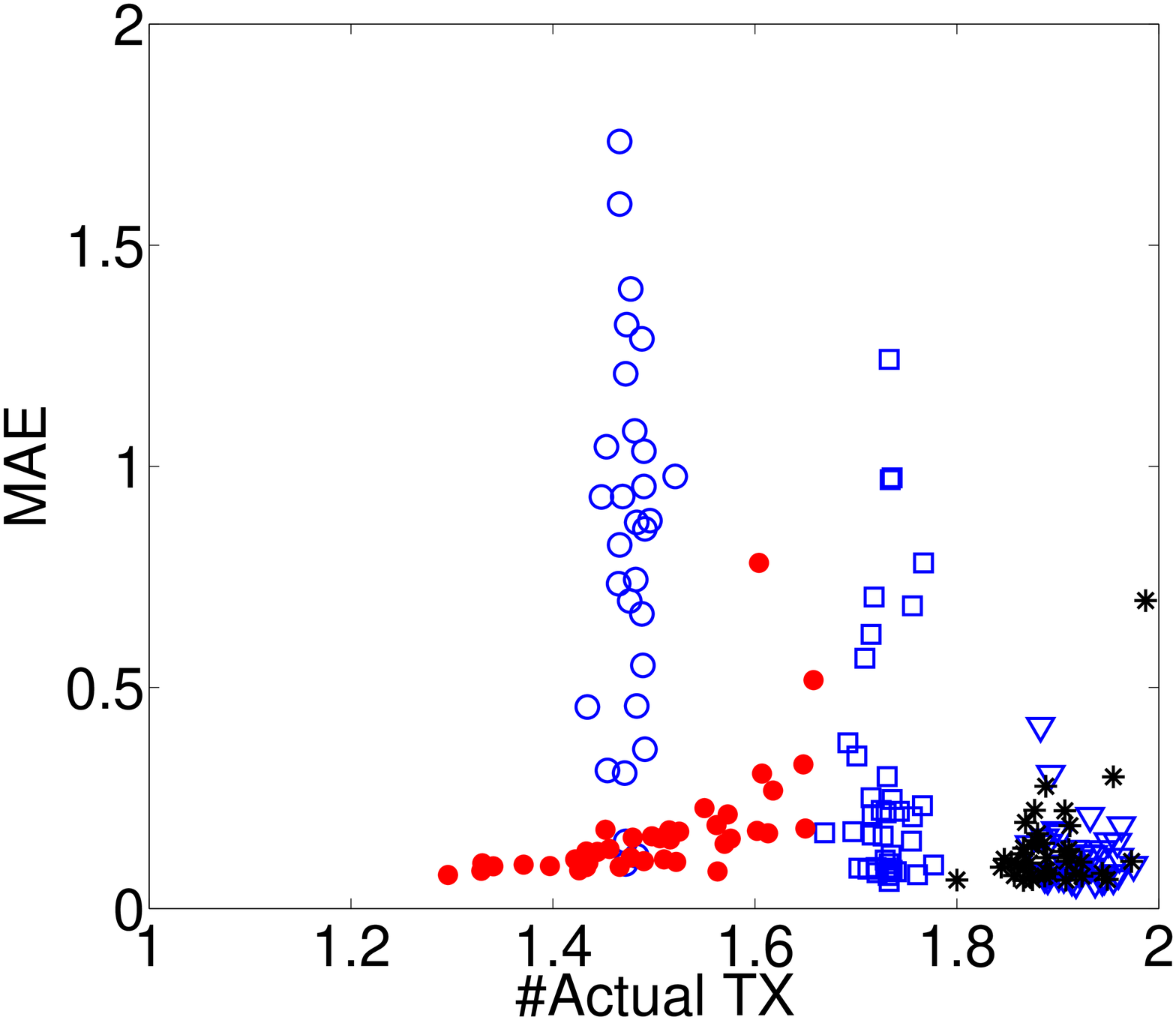}
    \subcaption{Relationship between MAE and $\#Actual\ TX$}
 \label{fig:spread_cy1}
 \end{minipage}
 \begin{minipage}[b]{0.45\textwidth}
\centering
  \includegraphics[width=1\columnwidth]{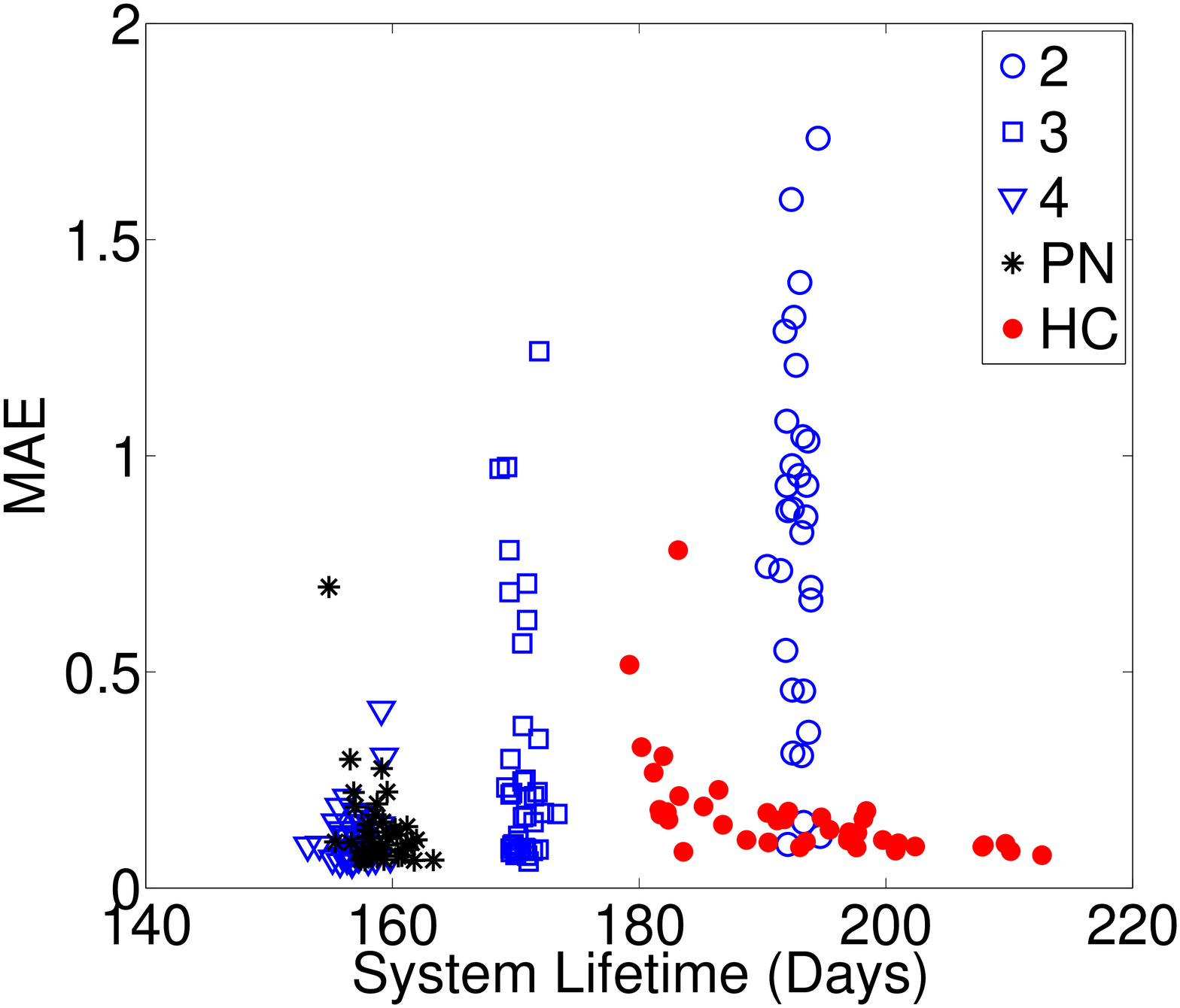}
    \subcaption{Relationship between MAE and battery life}
 \label{fig:spread_cy2}
 \end{minipage}
  \caption{%
Relationship between MAE and $\#Actual\ TX$ as well as system lifetime for different $\#TX$ adaptation algorithms under wireless interference. The algorithms are HC (red), PN (black) and fixed $\#TX$ (blue), respectively.
  }
  \label{fig:spread_cy}
\end{figure}

\reply{Fig.~\ref{fig:spread_cy} compares the relationship between MAE and the $\#Actual\ TX$, as well as the system lifetimes for different network configuration algorithms. HC's data points concentrate in the bottom left area of Fig.~\ref{fig:spread_cy1}, which indicates that this algorithm acquires smaller MAE with a fewer $\#Actual\ TX$. The simultaneous increase of both MAE and $\#Actual\ TX$ can be explained by the intuition of HC that poorer system performances will cause an increase of $\#Actual\ TX$; on the other hand, no extra transmissions will be adopted when the physical system is in a good condition. This trend indicates that network resources are adapted well based on the status of the physical plant. The same facts are reflected in Fig.~\ref{fig:spread_cy2}}.

%\note{Address there can be new scheduling algorithm and address scalability in order to better manage the network resources.}

\subsection{Simulation Under Sensor Disturbance}
\label{sec:eval_phy}

\begin{figure}[tp]
  \centering
  \includegraphics[width=0.88\columnwidth]{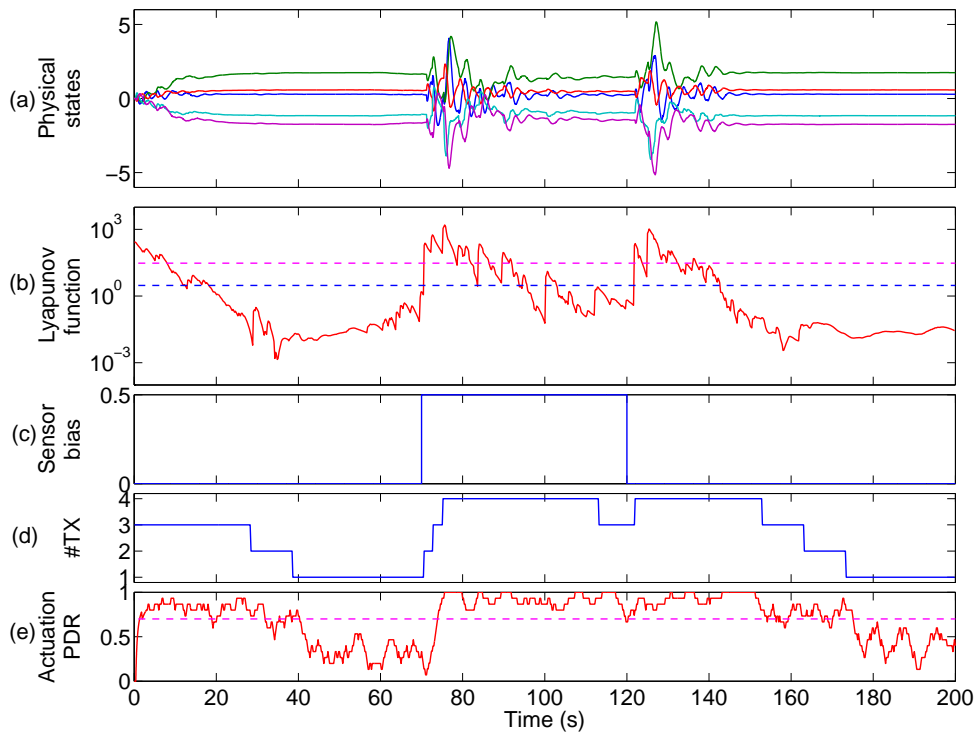}
  \caption{Statistics under physical disturbance}
  \label{fig:phymc}
\end{figure}

We now consider an external bias applied to the physical plant sensor, modeling either a malfunction or an adversarial attack.
That is, we consider a system where the output equation~\eqref{eq:diffeq} becomes $y_t = C\, x_t + \omega_t$ for $\omega_t \in \R$.
In this simulation, we keep the wireless background noise constant at $-76\, \text{dBm}$.

%\note{when there is sensor bias, the control command is more valuable than others, so more transmissions should be used. The ideas of packets are not of same importance.}

Fig.~\ref{fig:phymc} shows the result of our simulation, where a disturbance is applied in the interval $[70,120]\, \text{s}$, with $\omega_t = 0.5$.
Again, we focus our attention on two outcomes of this simulation.
First, twice our algorithm lowers the $\#TX$ to $1$, while maintaining the stability of the physical plant.
Note that the actuation packet delivery ratio in both situations is well below standard acceptable values, yet the use of an actuation buffer mitigates any significant impact that the information loss has in the physical plant.
Second, the plant remains stable under the application and later release of the physical disturbance, both times thanks to an increase in the $\#TX$.
This phenomenon validates our cyber-physical approach, where a controller that is designed to mitigate imperfections in the communication channels, together with a network manager that is aware of the performance of the physical plant, jointly result in an efficient and safe control system.

\begin{figure}[tp]
\begin{minipage}[b]{0.4\textwidth}
\centering
  \includegraphics[width=0.83\columnwidth]{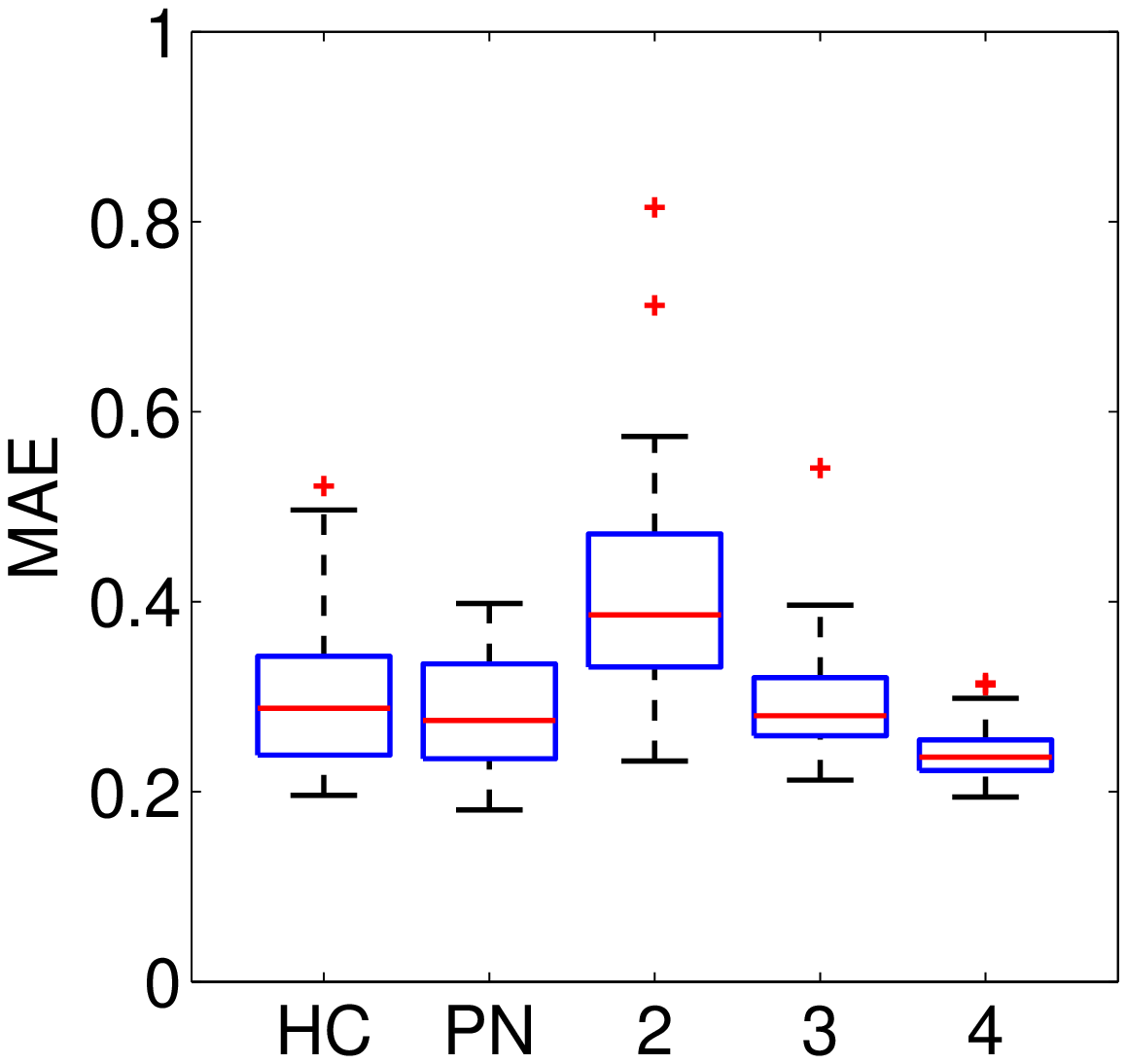}
 \subcaption{Mean Absolute Error (MAE)}
 \label{fig:phyattack1}
 \end{minipage}
 \begin{minipage}[b]{0.4\textwidth}
  \centering
  \includegraphics[width=0.8\columnwidth]{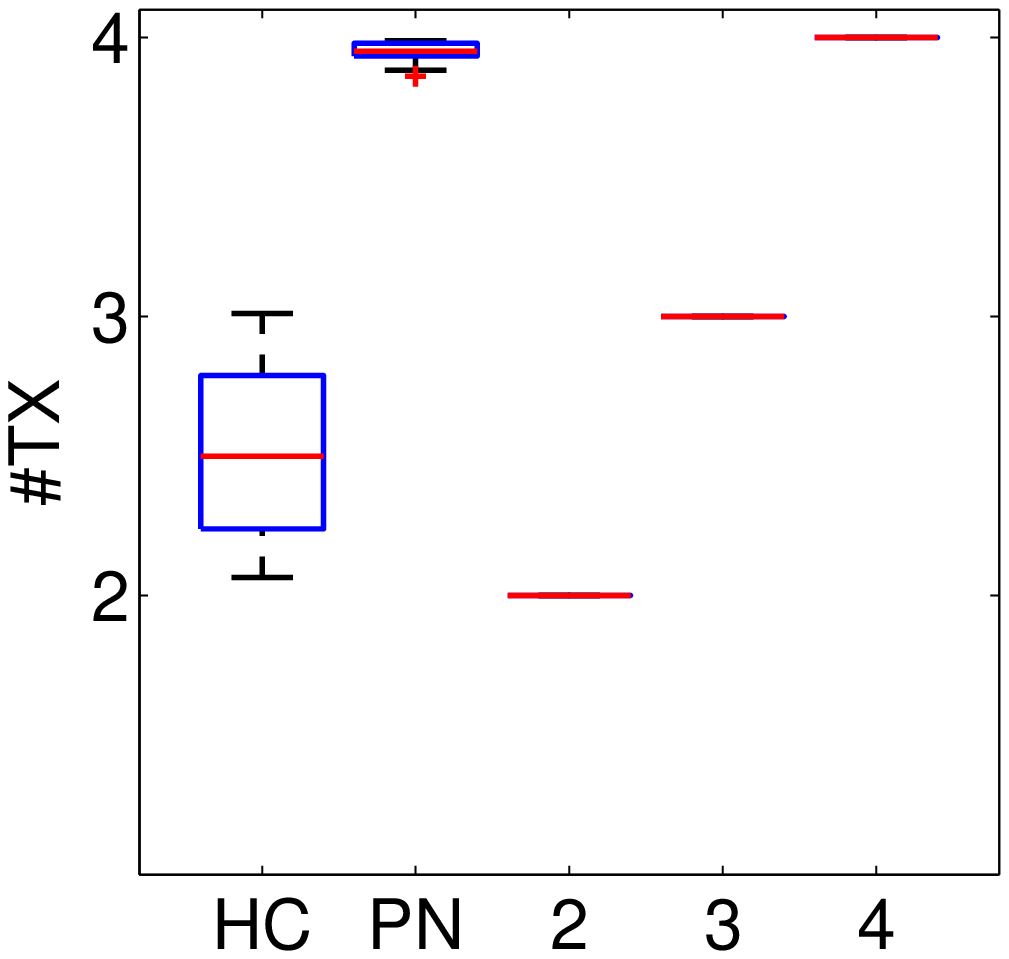}
   \subcaption{Average number of scheduled TXs}
    \label{fig:phyattack2}
 \end{minipage}\\
  \centering
\begin{minipage}[b]{0.4\textwidth}
\centering
  \includegraphics[width=0.85\columnwidth]{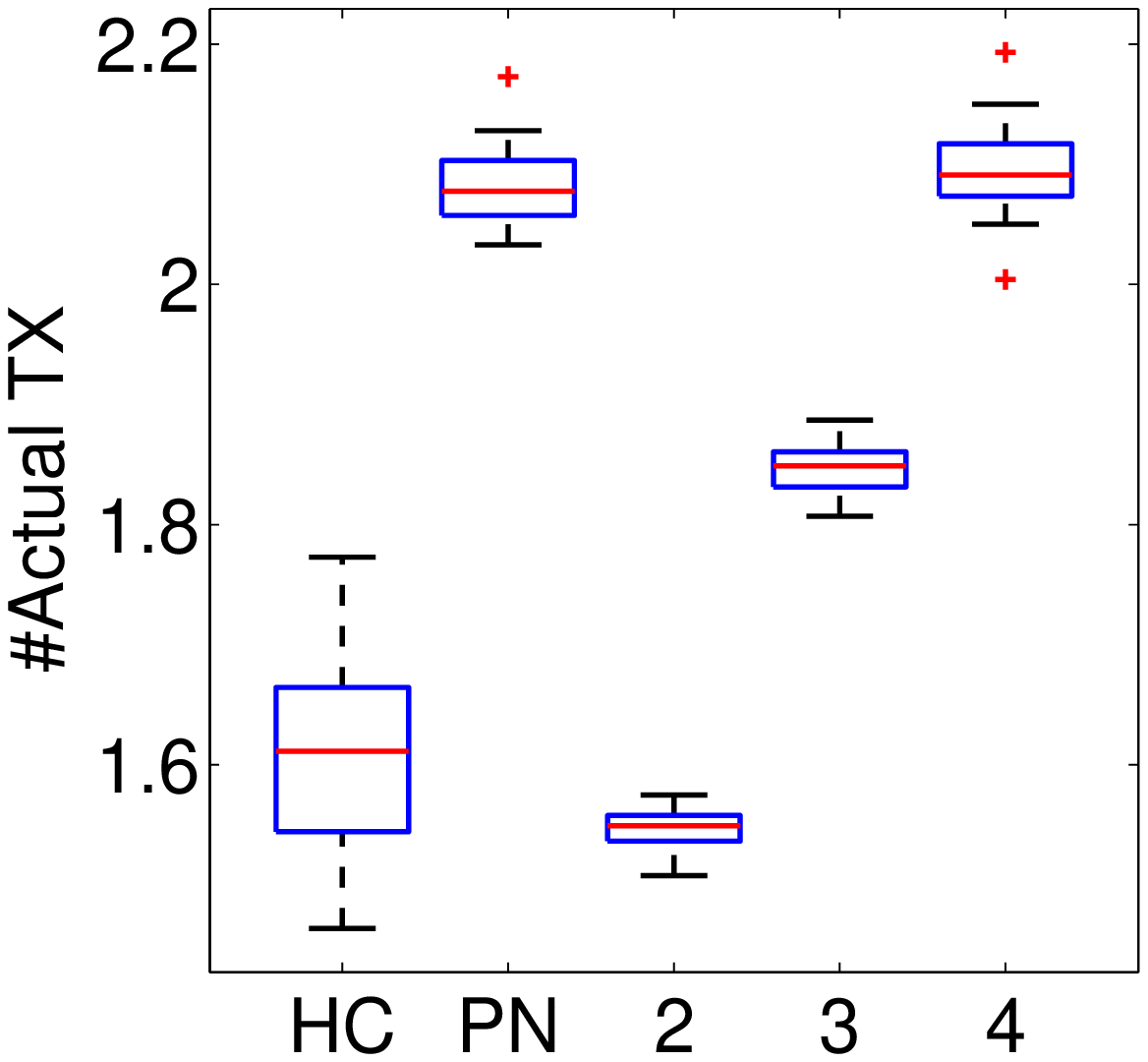}
   \subcaption{Average number of actual TXs per packet}
    \label{fig:phyattack3}
 \end{minipage}
\begin{minipage}[b]{0.4\textwidth}
 \centering
  \includegraphics[width=0.85\columnwidth]{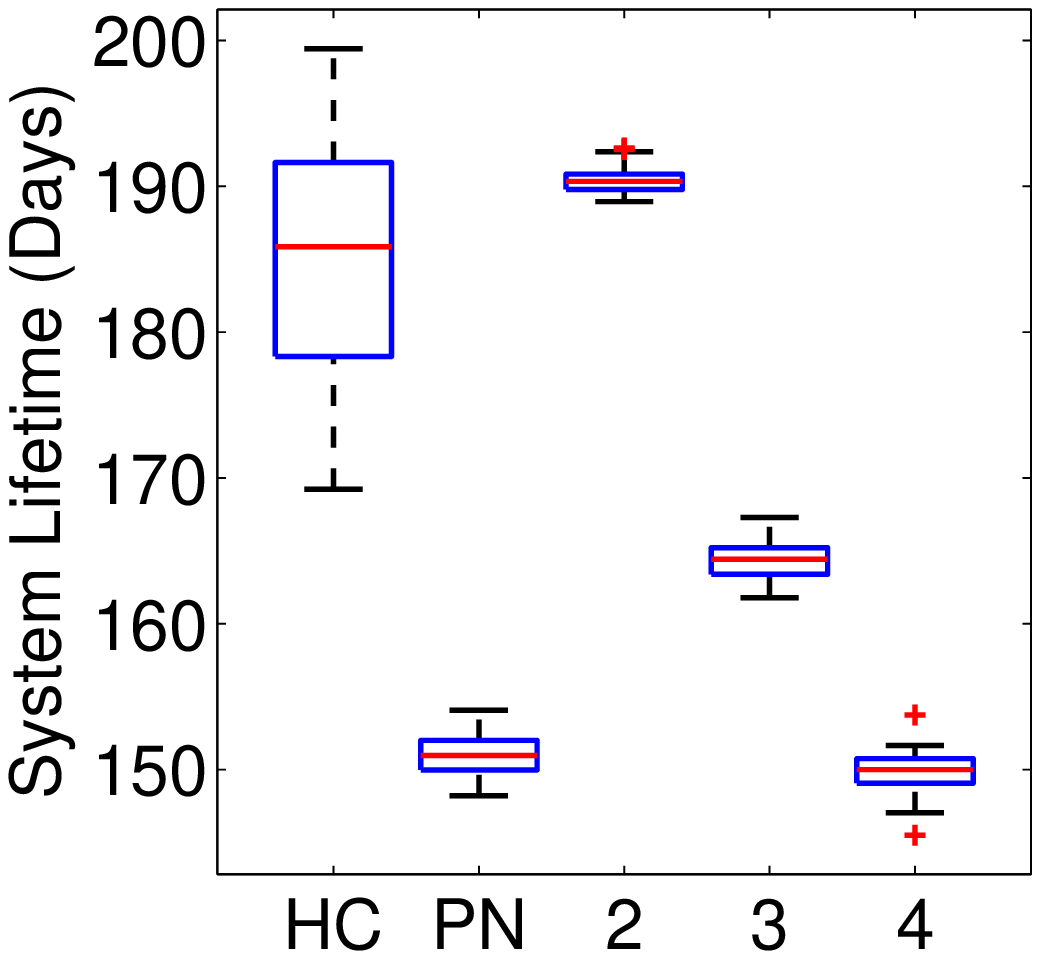}
   \subcaption{Battery life}
    \label{fig:phyattack4}
 \end{minipage}
  \caption{%
    (a) MAE, (b) average $\#TX$, (c) average $\#Actual\ TX$, and (d) battery life for different $\#TX$ adaptation algorithms under a sensor bias disturbance.
    The algorithms are our holistic controller (HC), pure network adaptation (PN), and constant transmission number equal to 2, 3, and 4, respectively.
  }
	\label{fig:phyattack}
\end{figure}

Fig.~\ref{fig:phyattack} is analogous to Fig.~\ref{fig:cyattack}, but it considers the physical disturbance described above.
Under sensor disturbance, fixed $2 TX$s results in 1~unstable simulation among 40, whereas none of the other algorithms produce unstable simulations.
Yet, HC results in executions with MAE comparable to constantly having a $\#TX$ equal to 3, as shown in Fig.~\ref{fig:phyattack1}, while achieving an average $\#TX$ of $2.5$, as shown in Fig.~\ref{fig:phyattack2}.
\reply{Fig.~\ref{fig:act_send_phy} shows the ratio for the $\#Actual\ TX$.
HC also obtains the highest ratio of $0$ and $1$ actual $TX$ per packet, since the system performances sometimes are acceptable even though the PDR is not high enough. Occasionally, HC increases the $\#Actual\ TX$ to $3$ and $4$ in an effort to handle physical attacks.
According to Fig.~\ref{fig:phyattack3} and Fig.~\ref{fig:phyattack4}, the $\#Actual\ TX$ of actuation WSAN nodes and the system lifetime of the WSAN are almost equal with the baseline of fixed 2 $TX$s.
Compared to PN, HC is significantly more efficient, validating our principle of choosing the network configuration based on the performance of the physical plant.}
\begin{figure}[tp]
\centering
  \includegraphics[width=0.5\columnwidth]{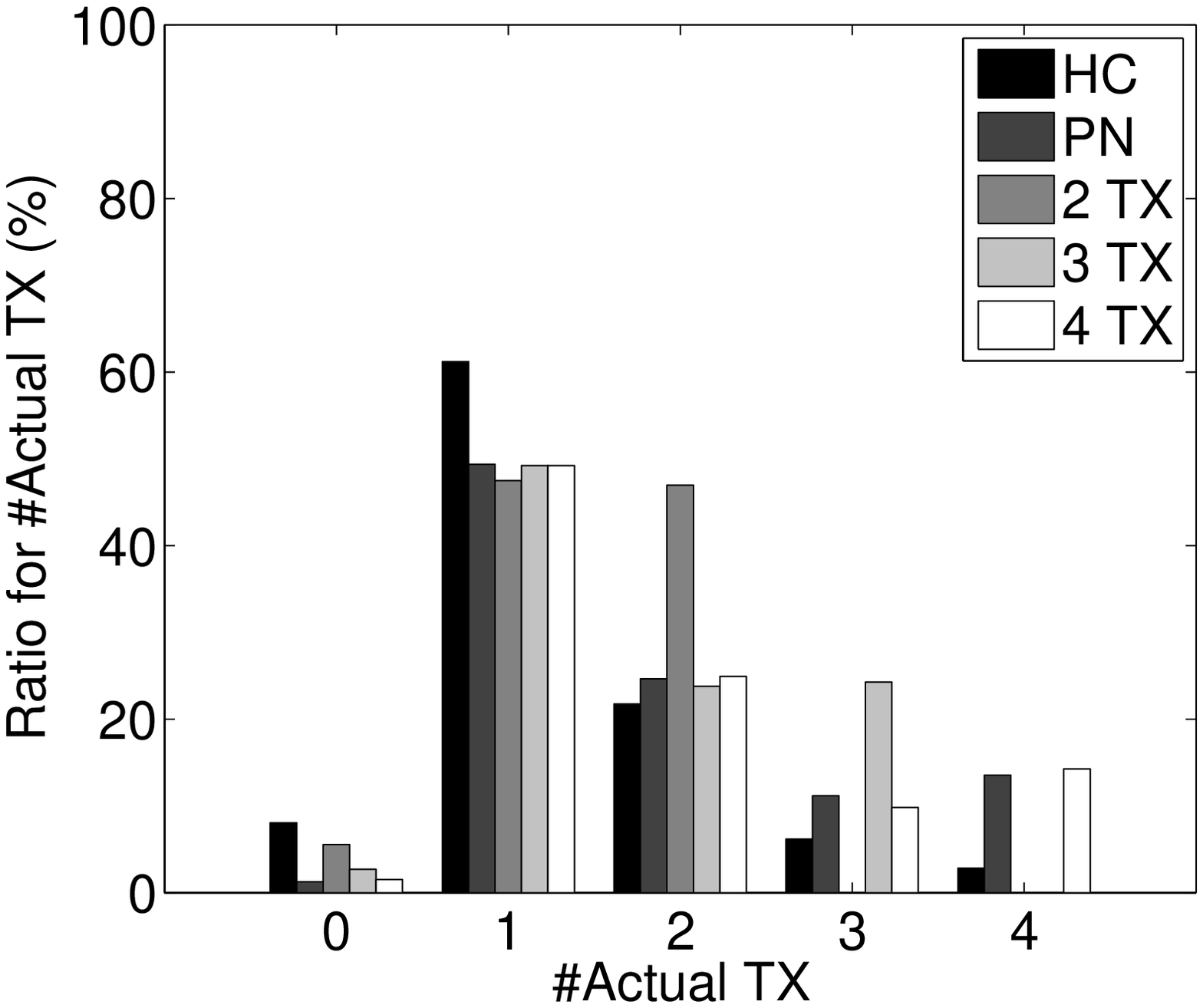}
  \caption{%
Ratio for $\#Actual\ TX$ under sensor bias disturbance
  }
  \label{fig:act_send_phy}
\end{figure}

\begin{figure}[tp]
\centering
\begin{minipage}[b]{0.45\textwidth}
\centering
  \includegraphics[width=1\columnwidth]{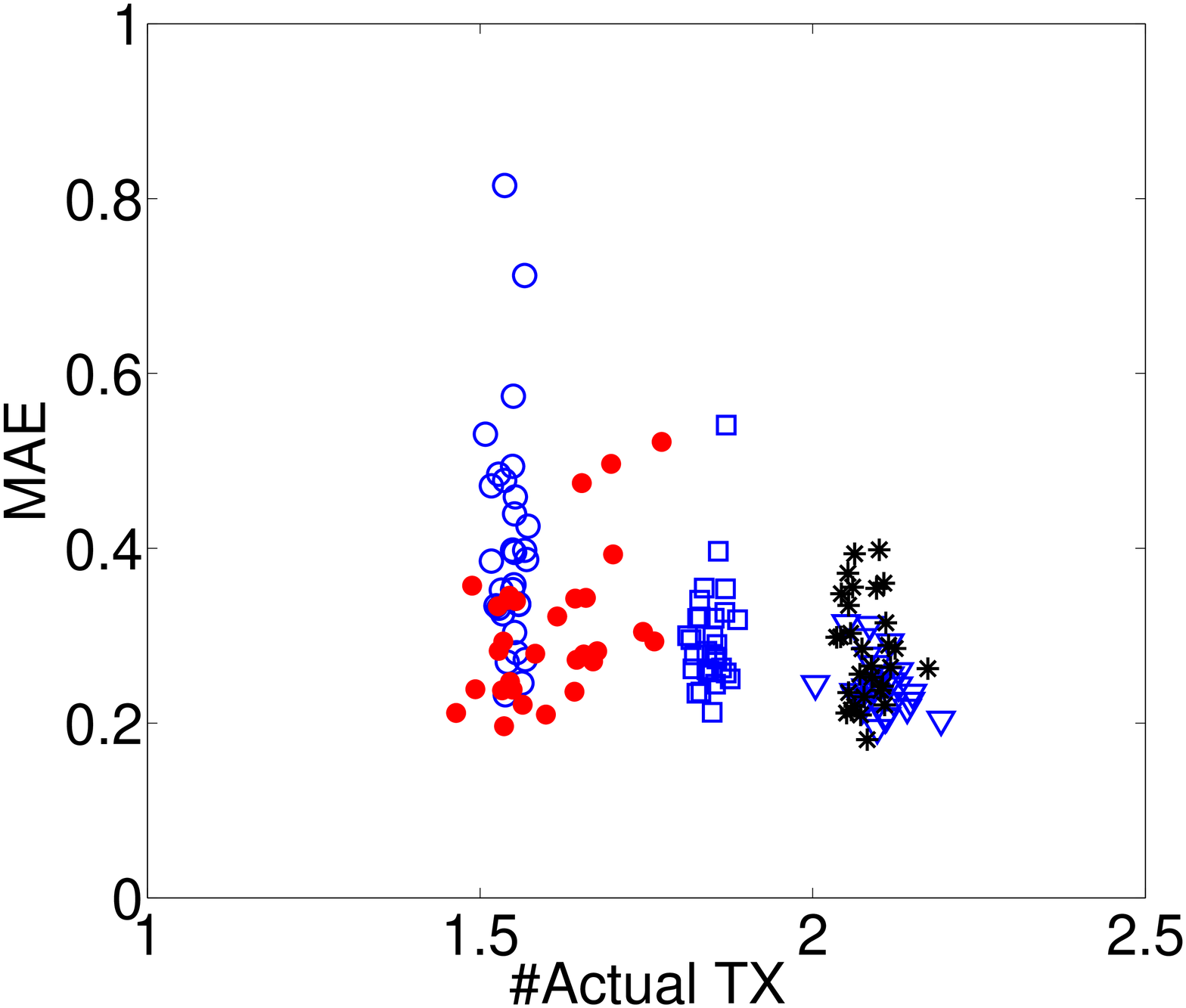}
     \subcaption{Relationship between MAE and $\#Actual\ TX$}
 \label{fig:spread_phy1}
 \end{minipage}
 \begin{minipage}[b]{0.45\textwidth}
\centering
  \includegraphics[width=1\columnwidth]{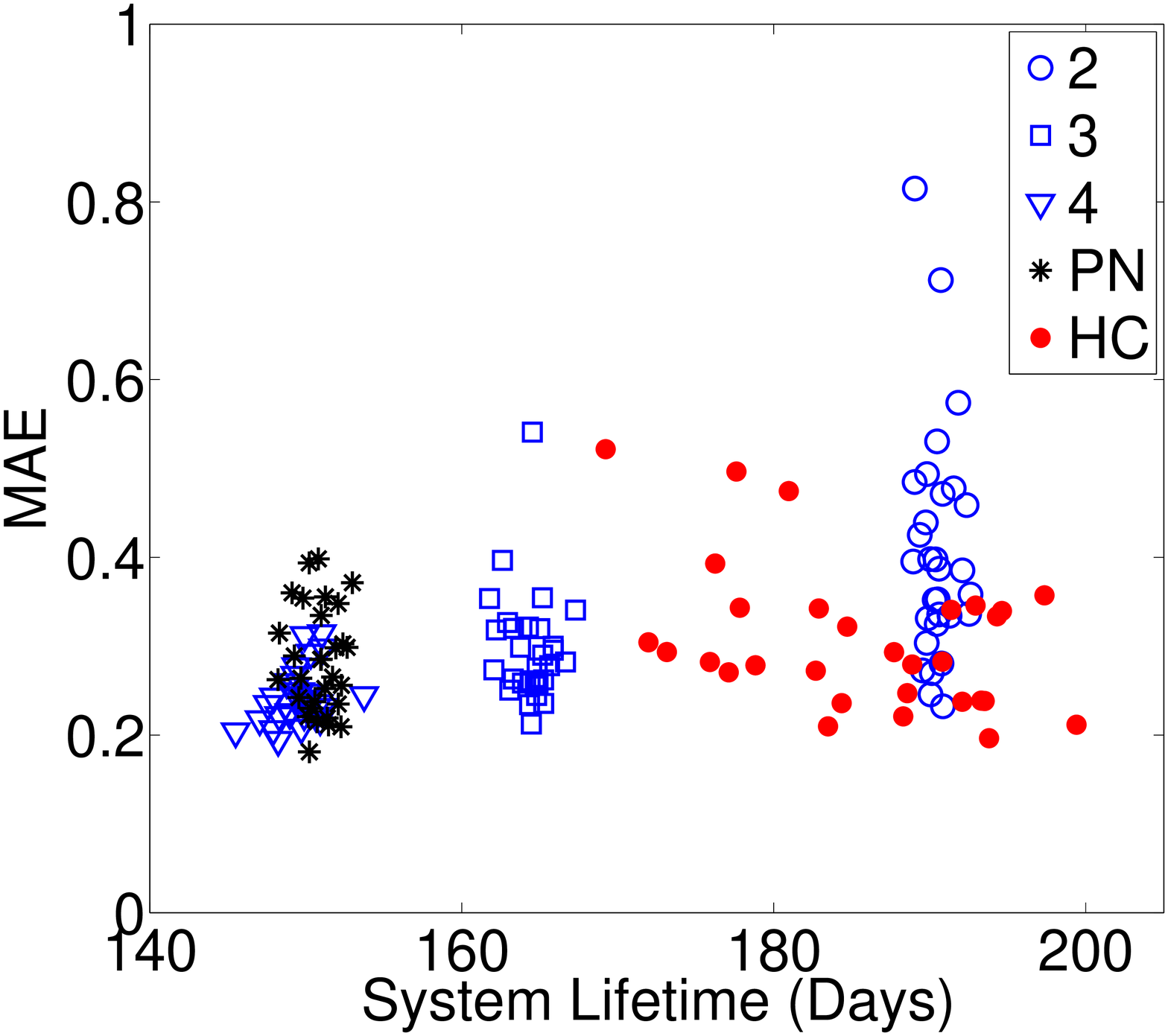}
    \subcaption{Relationship between MAE and battery life}
 \label{fig:spread_phy2}
 \end{minipage}
  \caption{%
Relationship between MAE and $\#Actual\ TX$ as well as system lifetime for different $\#TX$ adaptation algorithms under sensor bias disturbance. The algorithms are HC (red), PN (black) and fixed $\#TX$ (blue), respectively.
  }
  \label{fig:spread_phy}
\end{figure}

\reply{Fig.~\ref{fig:spread_phy} also presents similar results with Fig.~\ref{fig:spread_cy} in which the data points of HC concentrate in the bottom left area of Fig.~\ref{fig:spread_phy1} and the bottom right area of Fig.~\ref{fig:spread_phy2}. Therefore, HC allows a WSAN to be more resilient to both cyber and physical attack by adjusting its network configuration when needed.
It also extends the system lifetime while keeping MAE within reasonable values.}

%\note{Address there can be new scheduling algorithm and address scalability in order to better manage the network resources.}

%Therefore, the ReTX number is reduced to $0$ between $100s$--$300s$.
%At time $300s$, physical disturbances drag states away from the equilibrium point, leading cost function value to exceed its increase threshold.
%ReTX number increases accordingly during the transient process and reduces to $0$ after physical disturbance disappears and physical system comes back to equilibrium point again.

%We apply this physical disturbance case to all comparison methods.
%Again we compare IAE and ReTX of our Holistic Controller (HC) with Pure Network (PN) and fixed number of retransmissions.
%The comparison results are shown in Fig.~\ref{fig:phyattack}. It is worth noting that since the noise level is not as high as the cyber disturbances, states in all simulations converge in the end.
%We can see in Fig.~\ref{fig:phyIAE} that the IAE of Holistic Controller is similar with baseline of $3$ ReTX. While the average ReTX number of it is below $1.5$, as is shown in Fig.~\ref{fig:phyCost}.
%Besides, even though the protocol of pure network adaptation (PN) also get similar IAE as HC, PN concentrates on guarantee the PDR of network regardless the performances of physical system.
%As a result, the average ReTX number of PN is more than 2 times of HC.

\subsection{Multi-loop Control System}
\label{sec:multi}
In this section, we include an open-loop stable plant that shares the WSAN with the plant described in Section.~\ref{sec:setting} to form a multi-loop simulation. For this plant, the set of eigenvalues of A is equal to $\{0.4, 0.6, 0.96+0.02i, 0.96-0.02i, 0.8\}$. Since all eigenvalues are inside the unit circle, the plant is open-loop stable. The control loop of the plant described in Section.~\ref{sec:setting} uses the same sensing and actuation flows as in previous sections. The added control loop of the open-loop stable plant uses another pair of sensing and actuation flows in WSAN. \reply{Fig.~\ref{fig:multiloop_cy} shows the simulation results of added control loop under the same wireless interference as Section.~\ref{sec:eval_cyber}.

As is shown in Fig.~\ref{fig:multiloop_cy}d, during the transient state, the $\#TX$ is adjusted to $2$ for the first $23$ s.
Since the plant is open-loop stable, the Lyapunov function in Fig.~\ref{fig:multiloop_cy}b quickly decreases, and the $\#TX$ remains at $1$ after the transient state ends.
It is worth noting that, for open loop stable plants that stabilize at the equilibrium point, the system can retain perfect performances even when the network drops most of its packets.
As is shown in Fig.~\ref{fig:dlcyattack}, open-loop stable system can achieve smaller MAE with fewer $\#Actual\ TX$, which indicates that our holistic management framework is effective in allocating WSAN resources for the multi-loop system based on characteristics of physical plants .
In this case, it is possible to allocate those network resources to other lower priority applications such as network health reports.
However, the lower priority applications would have to be preempted if the actuation packets claim the network resources.
It is also possible to adopt other real-time scheduling and routing methods to balance the allocation of the network resources to multiple control loops based on their properties such as open-loop stability, time constant and run-time Lyapunov function.
We will address those issues in our future work.}
\begin{figure}[tp]
  \centering
  \includegraphics[width=0.88\columnwidth]{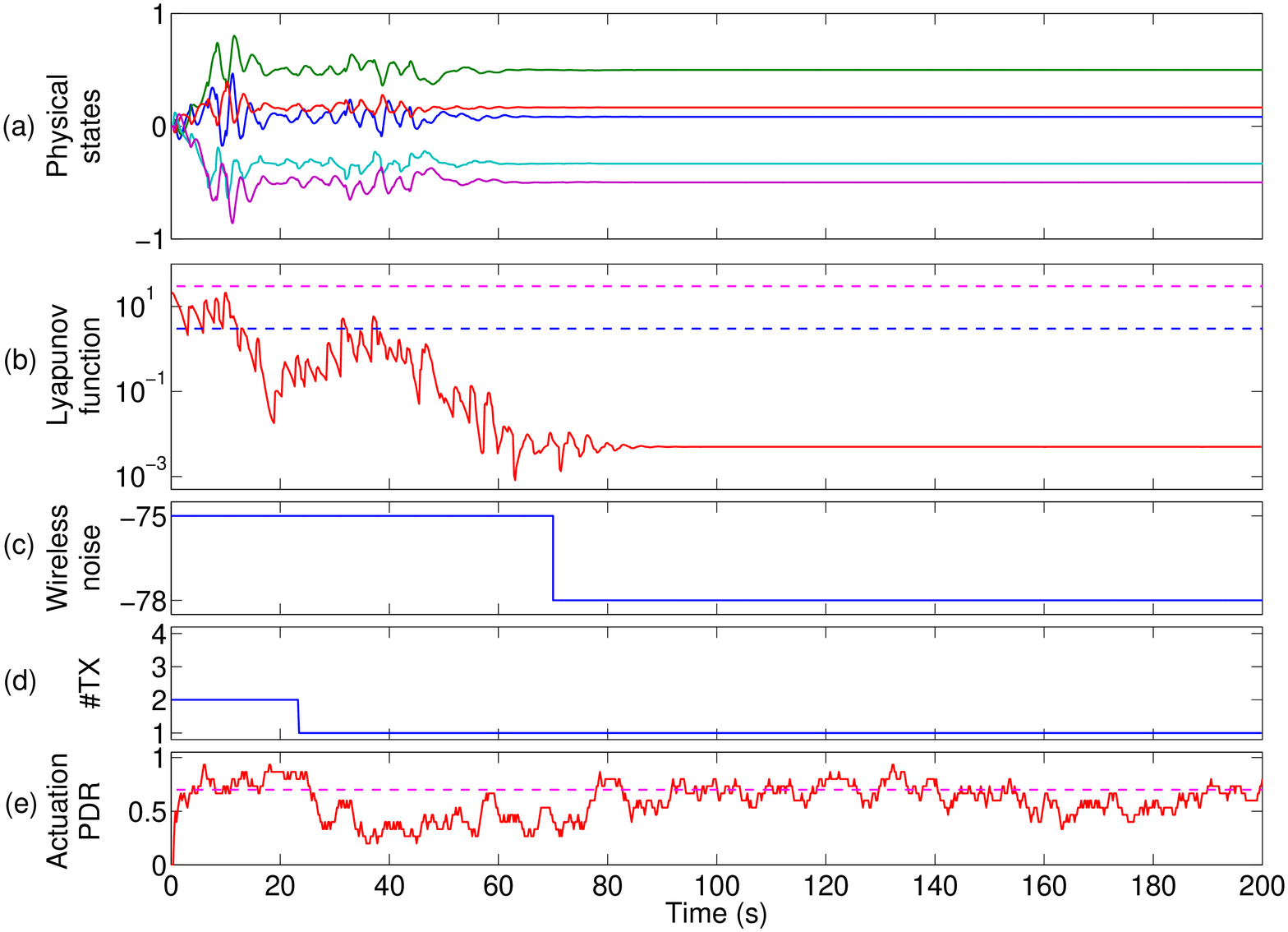}
  \caption{Open-loop stable plant under wireless disturbance}
  \label{fig:multiloop_cy}
\end{figure}

\begin{figure}[tp]
\begin{minipage}[b]{0.4\textwidth}
\centering
  \includegraphics[width=0.83\columnwidth]{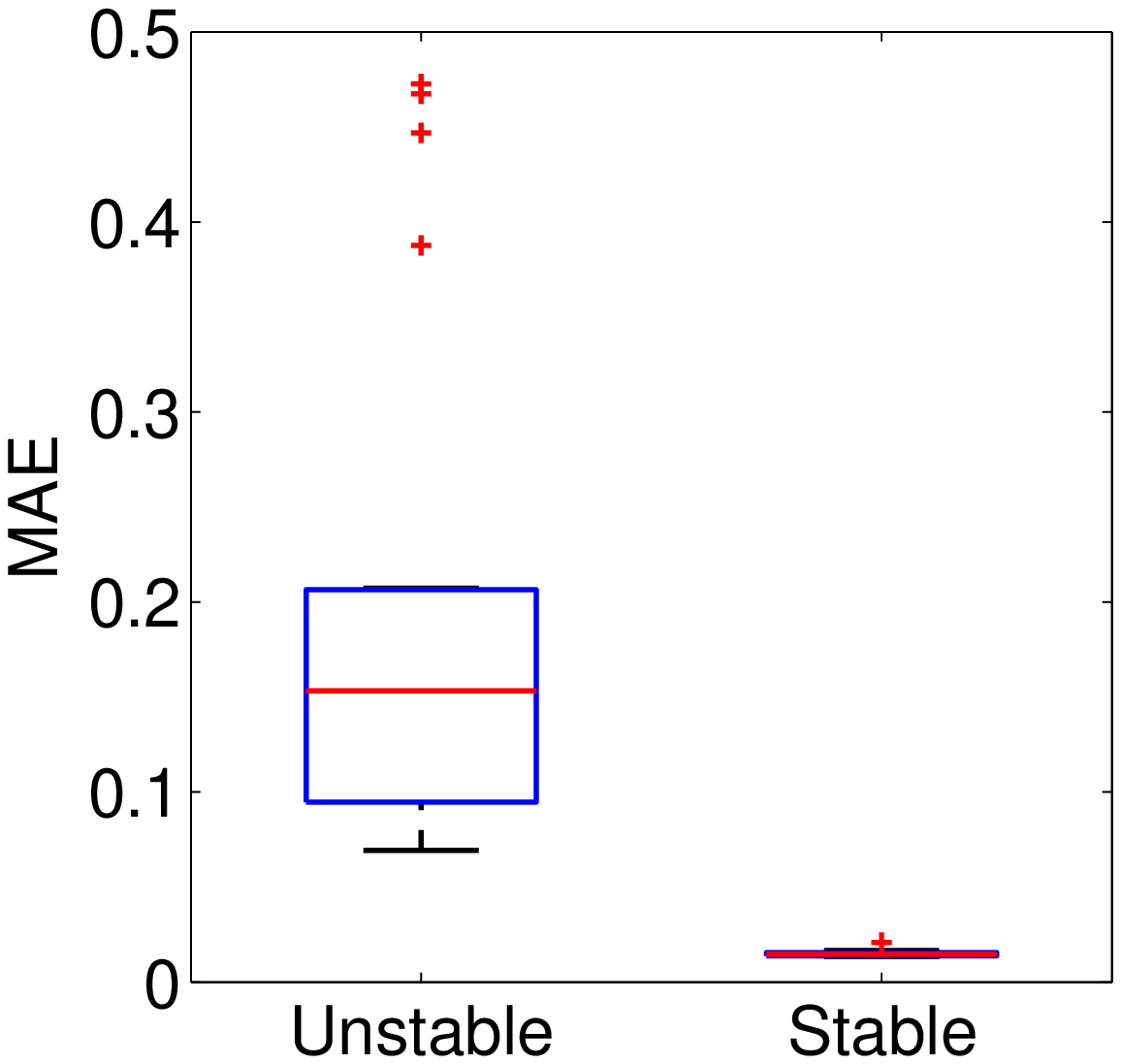}
 \subcaption{Mean Absolute Error (MAE)}
 \label{fig:dlattack1}
 \end{minipage}
 \begin{minipage}[b]{0.4\textwidth}
  \centering
  \includegraphics[width=0.85\columnwidth]{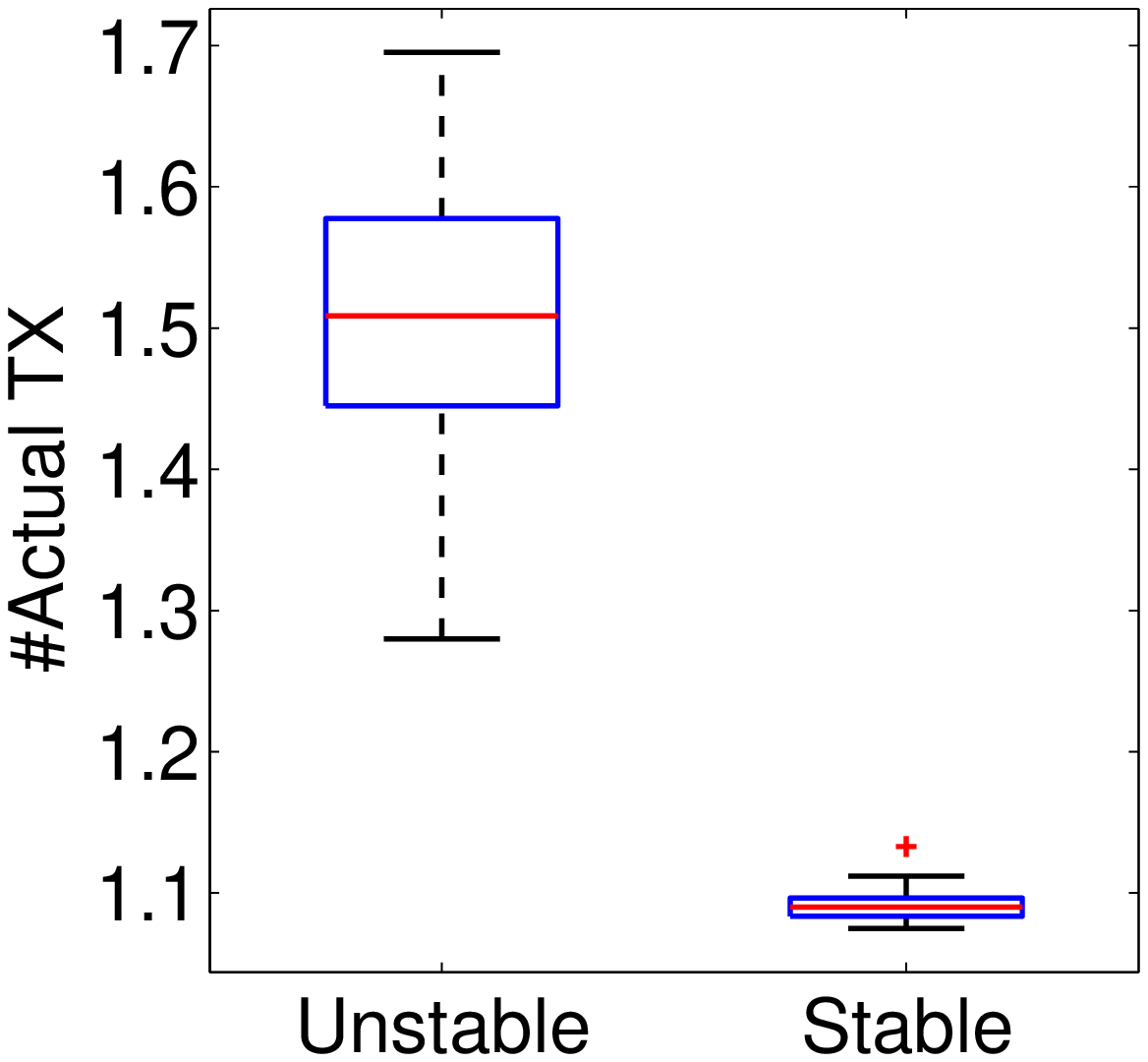}
   \subcaption{Average number of actual TXs per packet}
    \label{fig:dlattack2}
 \end{minipage}\\
   \caption{%
    (a) MAE, (b) $\#Actual\ TX$ of open-loop unstable and open-loop stable plants, respectively, under a wireless interferences. The $\#TX$ adaptation algorithm is HC.
  }
	\label{fig:dlcyattack}
\end{figure}

\section{Conclusions}
\label{sec:conclusion}

In this paper, we have proposed a holistic cyber-physical management framework to enhance the dependability of wireless control systems under both cyber and physical disturbances. The holistic management approach coordinate the physical control and network management mechanisms to safely control the physical plant while efficiently allocating wireless network resources.  We then design a concrete holistic controller that considers the worst-case evolution of the Lyapunov function of the plant under ideal network conditions, together with the run-time packet delivery ratio of the wireless network, and decides the number of transmissions for each wireless flow.  We have implemented the holistic controller and the network management mechanism in the WCPS simulator.  A case study that systematically explores both control and wireless performances has been presented using real-world wireless traces.  Simulation results show that our holistic controller is capable of maintain safe physical operation in the presence of sensor disturbances and significant wireless interference.  These results shed light on a new family of dependable cyber-physical systems that provides dependable control while efficiently allocating wireless network resources.

\appendix

\begin{acks}

This work is supported by NSF through grant 1646579 (CPS) and 1320921 (NeTS), and the Fullgraf Foundation.

\end{acks}

% Bibliography
\bibliographystyle{ACM-Reference-Format}
\bibliography{foo}

\end{document}